\documentclass[preprint,12pt]{elsarticle}
\usepackage{latexsym}
\usepackage{amsmath,amssymb,amsthm,mathtools}
\usepackage{amsfonts}
\usepackage{hyperref}
\usepackage[linesnumbered,lined,ruled,noend,vlined]{algorithm2e}
\usepackage{cleveref}
\usepackage{tabularx,environ}
\usepackage{comment}
\usepackage{url}
\usepackage[full]{complexity}
\usepackage{lscape}
\usepackage{lineno}
\usepackage{color}
\usepackage[normalem]{ulem}
\usepackage{mathrsfs}
\biboptions{sort&compress}

\usepackage{multirow}
\usepackage{booktabs}
\usepackage{caption} 
\usepackage[appendix=inline]{apxproof} 
\newtheoremrep{theorem}{Theorem}
\newtheoremrep{lemma}[theorem]{Lemma}
\newtheoremrep{proposition}[theorem]{Proposition}

\newcommand{\revise}[1]{{#1}}

\newcommand{\set}[1]{\{#1\}}
\newcommand{\inset}[2]{\{#1 \mid #2\}}
\newcommand{\size}[1]{\left|#1\right|}
\newcommand{\order}[1]{O(#1)}

\newtheorem{corollary}[theorem]{Corollary}
\newtheorem{observation}[theorem]{Observation}

\newcommand{\etal}{et al.\xspace}

\makeatletter
\@ifpackageloaded{algorithm2e}{
\SetKwProg{Fn}{Function}{}{}
\SetKwProg{Procedure}{Procedure}{}{}
\SetKwProg{Subprocedure}{Subprocedure}{}{}
\SetKwComment{tcc}{//}{}
\SetKwFunction{Output}{Output}%
\SetKw{Continue}{continue} 
\SetKwInOut{AlgInput}{Input}
\SetKwInOut{AlgOutput}{Output}
\SetKwInOut{AlgPrecondition}{Pre-conditions}
\SetKwInOut{AlgInvariant}{Invariants}
}{}
\makeatother

\makeatletter

\newcolumntype{\expand}{}
\long\@namedef{NC@rewrite@\string\expand}{\expandafter\NC@find}
\makeatletter

\newcommand{\problemtitle}[1]{\gdef\@problemtitle{#1}}
\newcommand{\probleminput}[1]{\gdef\@probleminput{#1}}
\newcommand{\problemoutput}[1]{\gdef\@problemoutput{#1}}
\NewEnviron{problem}{
  \problemtitle{}\probleminput{}\problemoutput{}
  \BODY
  \par\addvspace{.5\baselineskip}
  \noindent{
    \framebox[1.02\textwidth][c]{
        \begin{tabularx}{\textwidth}{@{\hspace{.2\parindent}} l X c}
            \multicolumn{2}{@{\hspace{.2\parindent}}l}{{\sc \@problemtitle}} \\
            \textbf{Input:}  & \@probleminput \\
            \textbf{Output:} & \@problemoutput
        \end{tabularx}
        }
      }
  \par\addvspace{.5\baselineskip}
}
\makeatother

\newcommand{\comp}[1]{{\mu}(#1)}

\journal{Discrete Applied Mathematics}
\begin{document}
\begin{frontmatter}

\title{Efficient Constant-Factor Approximate Enumeration of Minimal Subsets for Monotone Properties with Weight Constraints}

\author[1]{Yasuaki Kobayashi}
\ead{koba@ist.hokudai.ac.jp}
\author[2]{Kazuhiro Kurita}
\ead{kurita@i.nagoya-u.ac.jp}
\author[3]{Kunihiro Wasa}
\ead{wasa@hosei.ac.jp}

\address[1]{Hokkaido University, Sapporo, Japan}
\address[2]{Nagoya University, Nagoya, Japan}
\address[3]{Hosei University, Koganei, Japan}


\newcommand{\tw}{{\rm tw}}
\newcommand{\cw}{{\rm cw}}

\begin{abstract}
    A property $\Pi$ on a finite set $U$ is \emph{monotone} if for every $X \subseteq U$ satisfying $\Pi$, every superset $Y \subseteq U$ of $X$ also satisfies $\Pi$.
    Many combinatorial properties can be seen as monotone properties.
    The problem of finding a subset of $U$ satisfying $\Pi$ with the minimum weight is a central problem in combinatorial optimization.
    Although many approximate/exact algorithms have been developed to solve this kind of problem on numerous properties, a solution obtained by these algorithms is often unsuitable for real-world applications due to the difficulty of building accurate mathematical models on real-world problems.
    A promising approach to overcome this difficulty is to \emph{enumerate} multiple
    small solutions rather than to \emph{find} a single small solution.
    To this end, given a weight function $w: U \to \mathbb Q_{> 0}$ and $k \in \mathbb Q_{> 0}$, we devise algorithms that \emph{approximately} enumerate all minimal subsets of $U$ with weight at most~$k$ satisfying $\Pi$ for various monotone properties $\Pi$, where ``approximate enumeration'' means that algorithms output all minimal subsets satisfying $\Pi$ whose weight is at most~$k$ and may output some minimal subsets satisfying $\Pi$ whose weight exceeds $k$ but is at most $ck$ for some constant $c \ge 1$.
    These algorithms allow us to efficiently enumerate minimal vertex covers, minimal dominating sets in bounded degree graphs, minimal feedback vertex sets, minimal hitting sets in bounded rank hypergraphs, etc., of weight at most $k$ with constant approximation factors.
\end{abstract}

\end{frontmatter}

\section{Introduction}

Let $U$ be a finite set. A property $\Pi$ on $U$ is \emph{monotone} if for every $X \subseteq U$ satisfying $\Pi$, every superset $Y \subseteq U$ of $X$ also satisfies $\Pi$.
Many basic combinatorial properties, such as being a spanning subgraph of a graph, the linear dependency of columns of a matrix, and being a transversal of a hypergraph, can be seen as a monotone property on suitable sets $U$.
Thus, many combinatorial optimization problems can be formulated as a minimization problem with a monotone property over a finite set $U$.
Numerous minimization problems with monotone properties, such as the minimum vertex cover problem and the minimum dominating set problem, are proven to be intractable~\cite{garey1979computers}.
The concept of approximation algorithms is one of the most popular approaches to this difficulty, which aims to find a ``good'' solution with a provable guarantee on its weight.
However, in many real-world applications, such a ``good'' solution may be inadequate, even if one can find a best one, due to ambiguous ``true'' objectives and/or informal unwritten constraints.
To tackle these issues, enumerating multiple good solutions would be a promising approach.  
In this context, there are several attempts to enumerate multiple good solutions rather than to compute a single best solution, which have received considerable attention in the last decades.

One of the best-known attempts to achieve this goal is \emph{$K$-best enumeration}~\cite{Murty:Letter:1968,Lawler1972,Gabow:Two:1977,Eppstein:k-best:2016}.
An algorithm is called a \emph{$K$-best enumeration algorithm} if it generates $K$ distinct solutions and there is no solution (strictly) better than those generated by the algorithm. 
Obviously, the optimization counterpart corresponds to the case $K = 1$.
Many $K$-best enumeration algorithms have been developed in the literature: e.g., spanning trees~\cite{Gabow:Two:1977}, $s$-$t$ paths~\cite{Hoffman:Method:1959,Lawler1972}, $s$-$t$ cuts~\cite{Vazirani:Suboptimal:1992}, and (weighted) perfect matchings~\cite{Murty:Letter:1968} (see also~\cite{Eppstein:k-best:2016} for more information).  
Among others, enumerating multiple solutions with some specific order has attracted special interests in the field of database research, which they call \emph{ranked enumeration}~\cite{DBLP:conf/pods/RavidMK19,DBLP:conf/sigmod/TziavelisGR20,DBLP:journals/corr/abs-1902-02698}.
When we wish to enumerate $K$ solutions in the non-decreasing order of their quality, this is equivalent to $K$-best enumeration.

A common obstacle to applying these approaches to enumerate multiple (best) solutions with a monotone property $\Pi$ is that the underlying minimization problem has to be tractable since we need to find an optimal solution efficiently in the first place.
As in Table~\ref{tab:summary}, however, many problems of finding a minimum subset satisfying $\Pi$, which we focus on in this paper, are \NP-hard.
%
One of the possible solutions to avoid this obstacle is to consider approximation orders~\cite{DBLP:journals/jcss/FaginLN03,DBLP:conf/pods/KimelfeldS06,DBLP:conf/icde/AjamiC19}.
They relaxed the rigorous definition of the output order in ranked enumeration by considering an \emph{approximation order}.
More precisely, an enumeration algorithm outputs solutions in \emph{$\theta$-approximation order} for some $\theta \ge 1$ if, for every two solutions output by the algorithm, the cost of the former solution is not worse than $\theta$ times the cost of the latter solution.
This notion was first introduced by Fagin \etal~\cite{DBLP:journals/jcss/FaginLN03} and enables us to find an initial solution by polynomial-time approximation algorithms.
In particular, Ajami and Cohen~\cite{DBLP:conf/icde/AjamiC19} devised an algorithm for enumerating (weighted) set covers in $H_d$-approximation order in polynomial delay, where~$d$ is the largest size of a hyperedge and $H_d$ is the sum of the first $d$ numbers in the harmonic series.
Here, an enumeration algorithm runs in \emph{polynomial delay} if 
the followings are upper-bounded by a polynomial in the input size: 
(1) the maximum time interval between any two consecutive outputs, 
(2) the preprocessing time, and (3) the postprocessing time.     
However, the technique used in \cite{DBLP:conf/icde/AjamiC19} may produce \emph{non-minimal} solutions. 
These non-minimal solutions can be considered redundant in the sense that for a smallest solution $X$, there are exponentially many ``approximate'' solutions containing $X$, which can be easily obtained from $X$.
This redundancy highly affects the overall performance of enumeration algorithms. 
They extended their algorithm so that it enumerates only minimal solutions in approximation order.
However, the running time of this extension is no longer proven to be upper-bounded by a polynomial in the number of minimal solutions.
Our goal here is to develop enumeration algorithms such that the outputs satisfy the approximation quality requirement and the running time is upper-bounded by the number of solutions.

\subsection{Summary of our results}
In this paper, given a monotone property $\Pi$ on a finite set $U$, a weight function $w: U \to \mathbb Q_{>0}$, and $k \in \mathbb Q_{> 0}$, our aim is to design algorithms that enumerate all the \emph{inclusion-wise minimal} subsets $S$ of $U$ satisfying $\Pi$ whose weight $w(S)$ is at most $k$.
Here, we define $w(S) = \sum_{e \in S}w(e)$.
We measure the running time of these algorithms in an output-sensitive way (see \Cref{sec:prelim} and/or \cite{Johson:Yanakakis:Papadimitriou:IPL:1988,Strozecki:Enumeration:2019} for more information).
However, because our primary focus of this paper is on monotone properties whose minimization versions are \NP-hard, there seems to be no hope of developing even a polynomial-delay or incremental-polynomial time algorithm for enumerating subsets satisfying those properties with weight at most $k$.
Therefore, this paper introduces yet another concept to enumerate multiple good solutions.

Suppose that we are given a function $f: 2^U \to \mathbb{Q}_{>0}$ and $k \in \mathbb{Q}_{>0}$.
Let $\mathcal R$ be the collection of all feasible solutions in $2^{U}$ and $\mathcal S = \{S \in \mathcal R: f(S) \le k\}$.
An enumeration algorithm $\mathcal A$ \emph{approximately enumerates} $\mathcal S$ if it enumerates solutions $\mathcal R'$ with $\mathcal S \subseteq \mathcal R' \subseteq \mathcal R$ without duplication.
In other words, we allow enumeration algorithms to output some feasible solutions in $\mathcal R \setminus \mathcal S$ but forbid them to output infeasible solutions in $2^{U} \setminus \mathcal R$. 
The running time of an approximate enumeration algorithm $\mathcal A$ is measured by the input size and the cardinality of $\mathcal R'$.
Moreover, to evaluate the quality of solutions enumerated by $\mathcal A$, we define the \emph{approximation factor} of $\mathcal A$ as the ratio $\max_{S \in \mathcal R'}\frac{f(S)}{k}$. 
Given this, we call $\mathcal A$ an \emph{$\alpha$-approximate enumeration algorithm} when the approximation factor of $\mathcal A$ is at most $\alpha$. 

The definition of approximate enumeration algorithms {is natural} in a certain sense.
{
    Enumeration algorithms with $\theta$-approximation order mentioned in the previous subsection can be regarded as $\theta$-approximate enumeration algorithms.
    Let $(S_1, S_2, \ldots )$ be an output sequence of an enumeration algorithm with $\theta$-approximation order.
    By stopping the algorithm when it outputs a solution $S_i$ with the weight more than $k\theta$,
    the algorithm outputs all solutions with weight at most $k$ and does not output any solutions with weight more than $k\theta$, meaning that it is a $\theta$-approximate enumeration algorithm.
    Moreover, $K$-best enumeration algorithms can be regarded as $1$-approximate enumeration algorithms
    if the number $K$ of solutions of weight at most $k$ is known in advance.
    Notice that almost all of $K$-best enumeration algorithms in the survey~\cite{Eppstein:k-best:2016} output solutions in $1$-approximate order, that is, non-decreasing or non-increasing order.
    Thus, these algorithms can be regarded as $1$-approximate enumeration algorithms by just stopping execution when it outputs a solution of weight more than $k$ without using the value~$K$.}
Enumerating solutions with cardinality constraints are discussed in the field of parameterized algorithms~\cite{Fernau:parameterized:2002,Creignou:paradigms:2017}.
They focused on problems of enumerating solutions of cardinality at most $k$ and gave fixed-parameter tractable enumeration algorithms for those problems with respect to parameter $k$.
These algorithms are indeed considered as $1$-approximate enumeration algorithms in our setting.

In this paper, we develop two frameworks for designing efficient approximate enumeration algorithms for enumerating minimal subsets satisfying $\Pi$ for several monotone properties $\Pi$ with constant approximation factors.
The description of these frameworks is general, and hence we can derive approximate enumeration algorithms for many monotone properties in a unified way. 
We summarize some problems to which we can apply our frameworks and known approximation and enumeration (without weight constraints) results in \Cref{tab:summary}.
We would like to emphasize that our enumeration algorithms never output a non-minimal feasible solution even if its weight is at most $k$, {and work} efficiently even if $k$ is large in contrast to parameterized enumeration algorithms~\cite{Fernau:parameterized:2002,Creignou:paradigms:2017}.

We mainly focus on linear weight functions, that is, the weight of a solution is defined to be the sum of weights of elements in it.
However, our frameworks allow us to enumerate approximate solutions on more general weight functions.
For instance, we obtain a polynomial-delay $3$-approximate enumeration algorithm for minimal vertex covers with monotone submodular cost.

\newclass{\polydelay}{DelP}
\newclass{\incpoly}{IncP}
\newclass{\outputpoly}{OutP}

        \begin{table}[ht]
        \renewcommand{\baselinestretch}{0.9}
        \footnotesize
        \centering
        \def\arraystretch{1.2}
        \caption{Summary of approximation factors of known polynomial-time approximation algorithms for finding a \emph{minimum weight} $\Pi$-set, known results for enumerating \emph{minimal} $\Pi$-sets, and our results. 
        The approximation factors marked by $^*$ are for unweighted graphs, and these marked by $\dagger$ are obtained by reducing to the minimum weight $d$-hitting set problem. ``\polydelay'' and ``\incpoly'' stand for ``polynomial delay'' and ``incremental polynomial time'', respectively. }
        \begin{tabular}{ll|c|l|c|l}\hline
            \multicolumn{2}{c|}{Property $\Pi$}             & approx.                                                                                                           &  \multicolumn{1}{c|}{enum.}                                & \multicolumn{2}{c}{Our results}\\\hline
            \multirow{6}{*}{Sect. \ref{sec:framework}}            & {\sc Vertex Cover}                                                                                                & $2$                                                   & \polydelay~\cite{DBLP:journals/siamcomp/TsukiyamaIAS77}                          & $3$                        & \polydelay   \\
                                                            & {\sc Bounded Degree-$d$ Deletion}                                                                                 & $2 + \ln d$ \cite{OKUN2003231}                        & \polydelay~\cite{DBLP:journals/jcss/CohenKS08,DBLP:journals/corr/abs-2004-09885} & $3+\ln d$                  & \polydelay   \\
                                                            & {\sc Cluster Deletion}                                                                                            & $2$ {\cite{Aprile2023-nn}}                    & \polydelay~\cite{DBLP:journals/jcss/CohenKS08,DBLP:journals/corr/abs-2004-09885} & $3$                     & \polydelay   \\
                                                            & {\sc Split Deletion}                                                                                              & $2+\varepsilon$\cite{DBLP:conf/icalp/LokshtanovMPP020,Drescher:simple:2020} & \polydelay~\cite{DBLP:journals/jcss/CohenKS08,DBLP:journals/corr/abs-2004-09885} & $3 + \varepsilon$          & \polydelay   \\
                                                            & {\sc Pseudo Split Deletion}                                                                                       & $4^\dagger$                                           & \polydelay~\cite{DBLP:journals/jcss/CohenKS08,DBLP:journals/corr/abs-2004-09885} & $5$                        & \polydelay   \\
                                                            & {\sc Threshold Deletion}                                                                                          & $4^\dagger$                                           & \polydelay~\cite{DBLP:journals/jcss/CohenKS08,DBLP:journals/corr/abs-2004-09885} & $5$                        & \polydelay   \\
            \hline
            Thm. \ref{thm:approx-enum-sfd}                  & {\sc Star Forest Edge Deletion}                                                                                   & $3^\dagger$            & \polydelay                                                                       & $4$                        & \polydelay   \\
            \hline
            Thm. \ref{thm:approx-enum-dom}                  & \hspace{-0.2cm}\begin{tabular}{l}{\sc Dominating Set in}\\\hspace{1.6cm}{\sc Degree-$\Delta$ Graphs}\end{tabular} & $O(\log \Delta)$~\cite{Chvatal:Greedy:1979}           & \polydelay~\cite{Kante:Enumeration:2011}                                         & $O(\log \Delta)$           & \polydelay   \\
            \hline
            \multirow{5}{*}{Cor.~\ref{cor:enum-treewidth}}  & {\sc Star Forest Vertex Deletion}                                                                                 & $4^\dagger$            & \incpoly~\cite{DBLP:journals/corr/abs-2004-09885}                                        & $6$                        & \incpoly     \\
                                                            & {\sc Feedback Vertex Set}                                                                                         & 2\cite{Becker:Optimization:1996}                      & \polydelay~\cite{DBLP:journals/dam/SchwikowskiS02}                               & $4$                        & \incpoly     \\
                                                            & {\sc {Caterpillar-Forest Deletion}}                                                                                        & $7^*$ \cite{Philip:Quartic:2010}                        & \incpoly                                                                         & $9^*$                        & \incpoly     \\
                                                            & {\sc Planar-$\mathcal F$ Deletion}                                                                                & $O(1)^*$\cite{Fomin:Planar:2012}                        & \incpoly                                                                         & $O(1)^*$                     & \incpoly     \\
                                                            & {\sc Treewidth-$\eta$ Deletion}                                                                                   & $O(1)^*$\cite{Fomin:Planar:2012}                        & \incpoly                                                                         & $O(1)^*$                     & \incpoly     \\
            \hline
            \multirow{4}{*}{Cor.~\ref{cor:enum-cliquewidth}}& {\sc Block Deletion}                                                                                              & $4^*$ \cite{Agrawal:faster:2016}                        & \incpoly                                                                         & $6^*$                        & \incpoly     \\
                                                            & {\sc Bicluster Deletion}                                                                                          & $4^\dagger$                                           & \incpoly~\cite{DBLP:journals/jcss/CohenKS08} & $6$                        & \incpoly     \\

                                                            & {\sc Cograph Deletion}                                                                                            & $4^\dagger$                                           & \polydelay~\cite{brosse2020efficient}                                            & $6$                        & \incpoly     \\
                                                            & {\sc Trivially Perfect Deletion}                                                                                  & $4^\dagger$                                           & \polydelay~\cite{DBLP:journals/corr/abs-2004-09885}                              & $6$                        & \incpoly     \\
            \hline
            Thm. \ref{theo:dhit}                            & {\sc $d$-Hitting Set}                                                                                             & $d$ \cite{Bar-Yehuda:Linear:1981}  & \incpoly~\cite{DBLP:conf/latin/BorosEGK04}                                       & $\frac{(d + 4)(d - 1)}{2}$ & \incpoly  \\
             \hline
            Thm. \ref{theo:apx:eds}                         & {\sc Edge Dominating Set}                                                                                         & $2$\cite{Fujito:approximating:2017}             & \polydelay~\cite{DBLP:conf/wads/KanteLMNU15}                                     & $5$                        & \polydelay   \\
            \hline
            Thm.~\ref{theo:apx:st}                          & {\sc Steiner Subgraph}                                                                                            & $1.39$ \cite{Byrka:Steiner:2013}                      & \polydelay~\cite{DBLP:journals/is/KimelfeldS08}                                                                         & $2.39$                     & \polydelay     \\
            \hline
            \end{tabular}\label{tab:summary}
        \renewcommand{\baselinestretch}{1}
    \end{table}

As further algorithmic contributions, we give polynomial-delay approximate enumeration algorithms with constant approximation factors for two specific properties, {{\sc Edge Dominating Set}} and {\sc Steiner Subgraph}\footnote{In our setting, every minimal Steiner subgraph is a (minimal) Steiner tree and hence our approximate enumeration algorithm enumerates (minimal) Steiner trees only.}. 
For these properties, our framework for obtaining a polynomial-delay running time bound cannot be applied directly.
To obtain polynomial-delay bounds with constant approximation factors, {we instead present problem-specific techniques, which may be of independent interest.}

\subsection{Technical overview of our framework}
Although the concept of our approximate enumeration {is new}, the approach to this is rather well known.
We employ the \emph{supergraph technique}, which is frequently used in designing enumeration algorithms in the literature~\cite{DBLP:journals/jcss/CohenKS08,DBLP:conf/stoc/ConteU19,DBLP:journals/dam/SchwikowskiS02,DBLP:journals/algorithmica/KhachiyanBBEGM08,DBLP:journals/corr/abs-2004-09885}.
In this technique, we consider a directed graph $\mathcal G = (\mathcal V, \mathcal E)$ whose node set $\mathcal V$ corresponds to the set of solutions.
The key to designing an enumeration algorithm with the supergraph technique is to ensure that, by appropriately defining the arc set $\mathcal E$, $\mathcal G$ is strongly connected: If $\mathcal G$ is strongly connected, one can enumerate all the nodes of $\mathcal G$ by traversing its arcs from an arbitrary node.
{Cohen, Kimelfeld, and Sagiv~\cite{DBLP:journals/jcss/CohenKS08} invented a general technique to ensure the strong connectivity of $\mathcal G$ for enumerating maximal induced subgraphs with hereditary properties or connected hereditary properties.}
In their technique, they consider a subproblem of an enumeration problem, called an \emph{input-restricted problem}, and if one can efficiently solve this subproblem, they show that the entire enumeration problem can be solved efficiently as well.
In particular, if the input-restricted problem for a monotone property $\Pi$ can be solved in time polynomial in the input size, then we say that $\Pi$ has a \emph{CKS property}. Note that as they used this term for hereditary properties on graphs, we slightly abuse this notion for our purpose. 
{Recently}, Cao~\cite{DBLP:journals/corr/abs-2004-09885} comprehensively discussed the applicability of this technique. 

To apply this technique to our problem setting, we need to define a supergraph whose node set contains all minimal solutions of weight at most~$k$.
However, there are two obstacles to this: (1) we need to find a ``seed solution'' of weight at most~$k$, which is essentially difficult to obtain, and (2) we need to ensure the strong connectivity of the supergraph defined on the set of minimal solutions of weight at most $k$.
Fortunately, in our approximation setting, we do overcome these obstacles: (1) we can use known polynomial-time approximation algorithms for finding a ``seed solution'', and (2) for every two ``small'' minimal solutions, there is a path between these solutions whose ``internal solutions'' are all small in the supergraph defined by the weight constrained version of input-restricted problems.
Given this, it suffices to show that input-restricted problems with weight constraints can be solved efficiently.

More precisely, given a minimal subset $X$ of~$U$ satisfying $\Pi$ and $x \in X$, the input-restricted problem requires to enumerate all minimal subsets $R$ of $U$ with $w(R) \le k$ such that $(X \setminus \{x\}) \cup R$ satisfies $\Pi$.
Without the weight constraint on $R$, it is the original version of the input-restricted problem defined by Cohen \etal~\cite{DBLP:journals/jcss/CohenKS08}.
However, CKS properties are rather restrictive.
It is easy to observe that many monotone properties are not CKS properties.
In particular, the number of solutions of input-restricted problems can be exponential in the input size.
For such a property, Cohen \etal~\cite{DBLP:journals/jcss/CohenKS08} showed that if the solutions of the input-restricted problem can be enumerated in incremental polynomial time, the entire enumeration algorithm also runs in incremental polynomial time as well.
We develop approximate enumeration algorithms along this line.
To solve weight-constrained input-restricted problems, we can intensively exploit problem-specific structures.
For instance, to solve the {weight-constrained} input-restricted problem for minimal feedback vertex sets, it suffices to show that, given a vertex-weighted graph $G$, a minimal feedback vertex set $X$ of $G$, and $x \in X$, there is an incremental-polynomial time algorithm for enumerating minimal feedback vertex sets $R$ with $x \notin R$ and $w(R) \le k$ in the graph $G'$ obtained from $G$ by deleting $(X \setminus \{x\})$.
The key observation here is that $G'$ has treewidth at most two and hence we can easily enumerate such sets by dynamic programming over tree decompositions. 

For {\sc Edge Dominating Set} and {\sc Steiner Subgraph},
we exploit a new supergraph technique that ensures the reachability from a ``seed solution'' to every small solution by defining {an} elaborated neighborhood relation in $\mathcal G$, which is inspired by polynomial-delay enumeration algorithms for minimal induced trivially perfect graphs~\cite{DBLP:journals/corr/abs-2004-09885} and minimal induced Steiner subgraphs in claw-free graphs~\cite{kobayashi2020polynomial}.
Let us note that Kimenfeld and Sagiv~\cite{DBLP:conf/pods/KimelfeldS06} also gave a polynomial-delay algorithm for enumerating minimal Steiner trees with $2.39$-approximation order, which is the same approximation ratio of our approximate enumeration algorithm.
However, we believe that our algorithm is of some interest to readers as the same approximation ratio is achievable with our framework.

\section{Preliminaries}\label{sec:prelim}

Let $G = (V, E)$ be a graph without self-loops and multiple edges. 
The vertex set and edge set of $G$ are denoted by $V(G)$ and $E(G)$, respectively.
For a vertex $v \in V$, the neighborhood or the set of neighbors of $v$ is denoted by $N_G(v)$.
The set of incident edges of $v$ is denoted by $\Gamma_G(v)$. 
The maximum degree of a vertex in $G$ is denoted by $\Delta_G$, and we call it the maximum degree of $G$.
If the graph $G$ is clear from the context, we omit the subscript $G$ in these notations. 
A graph $H = (U, F)$ is a subgraph of $G$ if $U \subseteq V$ and $F \subseteq E$.
Moreover, $H$ is an induced subgraph of $G$ if $U \subseteq V$ and $F = \inset{\set{u, v} \in E}{u, v \in U}$. 
{
    We denote the induced subgraph $H$ as $G[U]$.
    For $U \subseteq V$, we denote the induced subgraph $G[V \setminus U]$ as $G - U$.
    Similarly, for $F \subseteq E$, we denote by $G - F$ the graph obtained from $G$ by deleting the edges in $F$, that is, $G - F = (V, E \setminus F)$.
}

Let $U$ be a finite set and let $\Pi$ be a monotone property over $U$.
Let $w: U \to \mathbb Q_{>0}$ be a weight function. 
For $X \subseteq U$, we define $w(X)$ as $\sum_{x \in X}w(x)$.
We call $w(X)$ the \emph{weight of $X$}.
A set $X \subseteq U$ is called a \emph{$\Pi$-set} if it satisfies $\Pi$.
We say that $X \subseteq U$ is a \emph{$(\Pi, w, k)$-set} if it is a $\Pi$-set and $w(X)$ is at most~$k$.
When $w$ is the unit weight function, we simply call it \emph{$(\Pi, k)$-set}.
A subset $X$ is a \emph{minimal $(\Pi, w, k)$-set} (resp. \emph{minimal $(\Pi, k)$-set}) if no proper subset $X'$ of $X$ is a $(\Pi, w, k)$-set (resp $(\Pi, k)$-set). 

We say that an enumeration algorithm runs in \emph{polynomial delay} if 
the following are {upper-bounded} by a polynomial in the input size $\size{U}$: (1) the maximum time interval between any two consecutive outputs, (2) the preprocessing time, and (3) the postprocessing time. 
An enumeration algorithm runs in \emph{incremental polynomial time} if the $i^\text{th}$ solution can be obtained in time $(i + |U|)^{O(1)}$.
In other words, given a set $\mathcal O \subseteq 2^{U}$ of solutions generated so far, the algorithm runs in time $(|\mathcal O| + |U|)^{O(1)}$ to find a new solution not contained in $\mathcal O$ or decide that every solution has already been contained in $\mathcal O$.
If the total running time of an enumeration algorithm is {upper-bounded} by a polynomial in both the input size and the number of outputs, we say that the algorithm runs in \emph{output-polynomial time} or \emph{polynomial total time}.

Now, we define our problem as follows: 

\begin{problem}
    \problemtitle{Minimal $(\Pi, w, k)$-Sets Enumeration Problem}
    \probleminput{A finite set $U$ with monotone property $\Pi$, a weight function $w: U \to \mathbb Q_{>0}$, and $k \in \mathbb Q_{> 0}$\revise{.}}
    \problemoutput{All minimal $(\Pi, w, k)$-sets of $U$\revise{.}}
\end{problem}

If we allow $k$ to be arbitrarily large, then {\sc Minimal $(\Pi, w, k)$-Sets Enumeration Problem} is equivalent to a usual enumeration problem for minimal subsets satisfying $\Pi$.
However, we note that the problem is particularly difficult in general as mentioned in the previous section. 
To cope with this difficulty, we allow us to use a ``seed set'', which is an arbitrary minimal $(\Pi, w, ck)$-set for some constant $c > 1$, as input.
For several monotone properties $\Pi$, such as {\sc Vertex Cover}, finding a minimum $\Pi$-set can be approximated within a constant factor.
Given such a ``seed set'', we devise a polynomial-delay algorithm for enumerating minimal $\Pi$-sets containing all the minimal $(\Pi, w, k)$-sets and some minimal $\Pi$-sets enumerated by the algorithm may have weight exceeding $k$. 
Recall that the approximation factor~$\alpha$ of the algorithm is defined by $\alpha = \frac{\tilde{k}}{k}$, where $\tilde{k}$ is the maximum weight of a solution output by the algorithm. 
Formally, our goal is to solve the {\sc Minimal $(\Pi, w, k)$-Sets Enumeration Problem} approximately.

\section{General frameworks for approximate enumeration}\label{sec:framework}

Our frameworks are based on the supergraph technique.
In this technique, we consider a directed graph $\mathcal G = (\mathcal V, \mathcal E)$ whose node set $\mathcal V$ corresponds to the collection of minimal $\Pi$-sets.
The key to designing an enumeration algorithm with the supergraph technique is to ensure that, by appropriately defining the arc set $\mathcal E$, $\mathcal G$ is strongly connected.
To ensure the strong connectivity of $\mathcal G$, we define a ``measure'' between two $\Pi$-sets. 
This idea has been used for the enumeration of maximal induced subgraphs with hereditary properties. {The property $\Pi$ is \emph{hereditary} if, for any graph satisfying $\Pi$, every induced subgraph also satisfies $\Pi$~\cite{DBLP:journals/corr/abs-2004-09885,DBLP:conf/stoc/ConteU19,DBLP:journals/jcss/CohenKS08,DBLP:journals/dam/SchwikowskiS02}.
Moreover, this idea has also been used for the enumeration of minimal subsets for specific properties~\cite{DBLP:journals/algorithmica/KhachiyanBBEGM08,DBLP:conf/ipco/BorosEGK04}.
The property $\Pi$ is \emph{connected hereditary} if, for any \revise{connected} graph satisfying $\Pi$, every connected induced subgraph also satisfies $\Pi$~\cite{DBLP:journals/jcss/CohenKS08}.\ Cohen \etal~\cite{DBLP:journals/jcss/CohenKS08} studied} enumeration problems for \emph{maximal} induced subgraphs with hereditary properties or connected hereditary properties.
The crucial observation is as follows: 
For any pair $H, H'$ of maximal induced subgraphs satisfying some hereditary property $\Pi$, 
if we can define a neighborhood of $H$ such that the neighborhood contains $H''$ with $|V(H) \cap V(H')| < |V(H'') \cap V(H')|$, then 
 $\mathcal G$ becomes strongly connected since $H$ has a neighbor $H''$ that is ``closer'' than $H$ to $H'$, which implies there is a directed path from $H$ to $H'$ in $\mathcal G$ via the neighborhood relation.

We apply the essentially same idea to {\sc Minimal $(\Pi, w, k)$-Sets Enumeration Problem}.
For a minimal $\Pi$-set $S \subseteq U$, we develop an enumeration algorithm for enumerating minimal $\Pi$-sets $S_1, S_2, \ldots, S_t \subseteq U$ such that for every minimal $\Pi$-set $S'$, there is $S_i$ with $\size{S \cup S'} > \size{S_i \cup S'}$.
The key difference from \cite{DBLP:journals/jcss/CohenKS08} is that we use the cardinality of the union with $S'$ as a ``measure'' rather than that of the intersection, which is also important for approximation guarantees.
{
    The following observation is easy but essential.
    \begin{observation}\label{lem:union}
        Let $S$ and $S'$ be minimal $\Pi$-sets. 
        Then, $S = S'$ if and only if $\size{S} = \size{S \cup S'}$. 
    \end{observation}
}

To design a $c$-approximation enumeration algorithm for {\sc Minimal $(\Pi, w, k)$-Sets Enumeration Problem}, 
our strategy is to construct a subgraph $\mathcal G'$ of $\mathcal G$ such that 
\begin{itemize}
    \item $\mathcal G'$ contains all the minimal $(\Pi, w, k)$-sets,
    \item \revise{each vertex in $\mathcal G'$ corresponds to a minimal $(\Pi, w, ck)$-set, and}
    \item from an arbitrary minimal $(\Pi, w, ck)$-set $S$ in $\mathcal G'$, there is a directed path to every minimal $(\Pi, w, ck)$-set in $\mathcal G'$,
\end{itemize}
where $c$ is a value depending on the property $\Pi$ and a known polynomial-time approximation algorithm for finding seed solution $S$.
This can be done by solving a weight-constrained version of the input-restricted problem, which is defined as follows.

\newcommand{\IRP}[3]{\textsc{$(#1, #2, #3)$-IRP}\xspace}

\begin{problem}
    \problemtitle{Input-Restricted Minimal $(\Pi, w, k)$-Sets Problem (\IRP{\Pi}{w}{k})}    
    \probleminput{A finite \revise{set} $U$ with monotone property $\Pi$, a weight function $w: U \to \mathbb Q_{> 0}$, $k \in \mathbb Q_{> 0}$, a minimal $\Pi$-set $X \subseteq U$, and an element $x \in X$.}
    \problemoutput{All minimal subsets $R \subseteq U \setminus \set{x}$ with $w(R) \le k$ such that $(X \setminus \set{x}) \cup R$ is a $\Pi$-set. }
\end{problem}
{The solution} of \IRP{\Pi}{w}{k} is denoted by $\mathcal S(U, \Pi, w, k, X, x)$.
Since \IRP{\Pi}{w}{k} is also an enumeration problem of a monotone property with a weight constraint, we can apply approximate enumeration algorithms to it.
Let $\mathcal A$ be an enumeration algorithm for \IRP{\Pi}{w}{k} with approximation factor $d$. 
The collection of sets output by $\mathcal A$ is denoted by $\mathcal S_{\mathcal A}(U, \Pi, w, k, X, x)$, that is, $\mathcal S_{\mathcal A}(U, \Pi, w, k, X, x)$ contains all inclusion-wise minimal $R \subseteq U$ with $w(R) \le k$ such that $(X \setminus \set{x}) \cup R$ is a $\Pi$-set and may contain some inclusion-wise minimal $R \subseteq U$ with $k < w(R) \le dk$ such that $(X \setminus \set{x}) \cup R$ is a $\Pi$-set. 
Note that for $R \in \mathcal S_{\mathcal A}(U, \Pi, w, k, X, x)$, $(X \setminus \{x\}) \cup R$ may not be a minimal $\Pi$-set.
If no confusion arises, we simply write $\mathcal S$ (resp. $\mathcal S_{\mathcal A}$) to denote $\mathcal S(U, \Pi, w, k, X, x)$ (resp. $\mathcal S_{\mathcal A}(U, \Pi, w, k, X, x)$).

A function $\comp{\tilde{X}}$ computes a minimal $\Pi$-set from a $\Pi$-set $\tilde{X}$ by greedily removing an element $e$ from $\tilde{X}$ as long as $\tilde{X}\setminus \set{e}$ is a $\Pi$-set. 
Note the resultant set of $\comp{\tilde{X}}$ is a subset of $\tilde{X}$ and 
if the membership of $\Pi$ can be checked in polynomial time, it is computed in polynomial time as well. 
The following lemmas show basic properties of $\mathcal{S}_{\mathcal A}$ and $\comp{\cdot}$. 

\begin{lemmarep}\label{lem:comp:uniqueness}
    Let $X \subseteq U$ be a minimal $\Pi$-set and let $x \in X$.
    For distinct $R, R' \in \mathcal S_{\mathcal A}$,
    $\comp{(X \setminus \{x\}) \cup R} \neq \comp{(X \setminus \{x\}) \cup R'}$.
\end{lemmarep}
\begin{proof}
    It is sufficient to prove that $R \subseteq \comp{(X \setminus \{x\}) \cup R}$ for $R \in \mathcal S_{\mathcal A}$.
    Suppose for contradiction that $R \cap \comp{(X \setminus \set{x}) \cup R} = T$ for some $T \subset R$.
    By the monotonicity of $\Pi$, $(X \setminus \{x\}) \cup T$ is a $\Pi$-set, contradicting to the minimality of $R$.
\end{proof}

\begin{lemmarep}\label{lem:input-restricted}
    Let $X \subseteq U$ be a minimal $\Pi$-set and let $x \in X$.
    Then, for every minimal $(\Pi, w, k)$-set $Y\subseteq U \setminus \set{x}$ excluding $x$, there is $R \in \mathcal S_{\mathcal A}$ such that $\comp{(X \setminus \set{x}) \cup R} \cup Y \subset X \cup Y$. 
\end{lemmarep}
\begin{proof}
    By the monotonicity of $\Pi$, $(X \setminus \set{x}) \cup Y$ is a $\Pi$-set. 
    Then, there is an inclusion-wise minimal set $R \subseteq Y$ such that $(X \setminus \set{x}) \cup R$ is a $\Pi$-set and $w(R) \le w(Y) \le k$, which implies that $R \in \mathcal S_{\mathcal A}$.
    Since $R \subseteq Y$ and $x \notin R$, we have
    \begin{linenomath}
        \begin{align*}
            X \cup Y \supset (X \setminus \{x\} \cup R) \cup Y \supseteq \comp{(X \setminus \set{x}) \cup R} \cup Y,
        \end{align*}
    \end{linenomath}
    completing the proof of the lemma.
\end{proof}

We define an arc from a minimal $\Pi$-set $X$ to a minimal $\Pi$-set $Z$ if $Z = \comp{(X \setminus \{x\}) \cup R}$ for some $x \in X$ and $R \in \mathcal S_{\mathcal A}$.
By \Cref{lem:comp:uniqueness,lem:input-restricted}, we can conclude that $\mathcal G$ has a directed path from every minimal $\Pi$-set $X$ to every minimal $(\Pi, w, k)$-set $Y$.
Moreover, every internal node $Z$ on the path is a minimal $(\Pi, w, w(X \cup Y))$-set since $Z \cup Y \subset X \cup Y$.
    We show that an approximate enumeration algorithm can be obtained by traversing a portion of $\mathcal G$.
    We give the details of the algorithm in \Cref{alg:general}.

\begin{algorithm}[t]
    \SetKwInput{KwInput}{Input}
    \SetKwInput{KwOutput}{Output}
    \SetAlgoLined
    \KwInput{A minimal $\Pi$-set $S$ of $U$ with $w(S) \le ck$ for some $c \ge 1$}
    Output $S$ and add $S$ to $\mathcal Q$ and $\mathcal O$\;
    \While{$\mathcal{Q}$ is not empty}{
        Let $X$ be a minimal $\Pi$-set in $\mathcal Q$ and delete {$X$} from $\mathcal Q$\; 
        \ForEach{$x \in X$}{ \label{alg:general:comp:neighbors}
            Compute $\mathcal S_{\mathcal A}(U, \Pi, w, k, X, x)$ by a $d$-approximate enumeration algorithm $\mathcal A$\; 
            \ForEach{$R \in \mathcal S_{\mathcal A}(U, \Pi, w, k, X, x)$}{
                $Z \gets \comp{(X \setminus \set{x}) \cup R}$\; 
                \If{$Z \notin \mathcal O$}{
                    Output $Z$ and add $Z$ to $\mathcal O$\; \label{alg:general:output:X}
                    \If{$w(Z) \le w(S) + k$}{
                    Add $Z$ to $\mathcal Q$\; 
                    } 
                }
            }
        }
    }
    \caption{$(c + d + 1)$-approximate enumeration for {\sc Minimal $(\Pi, w, k)$-Sets Enumeration Problem.} }
    \label{alg:general}
\end{algorithm}

\begin{theoremrep}\label{theo:apx:output}
    Let $U$ be a finite set, let $\Pi$ be a monotone property over $U$, let $w$ be a weight function, and let $k$ be positive.
    Suppose that the membership of $\Pi$ can be decided in polynomial time and there is an output-polynomial time enumeration algorithm $\mathcal A$ for \IRP{\Pi}{w}{k} with approximation factor $d$. 
    Then, given a minimal $\Pi$-set of weight at most $ck$, \Cref{alg:general} solves {\sc Minimal $(\Pi, w, k)$-Sets Enumeration Problem} with approximation factor $c + d + 1$ in output polynomial time.
    Moreover, {\Cref{alg:general} runs in incremental polynomial time} if $\mathcal A$ runs in incremental polynomial time.
\end{theoremrep}

\begin{proof}
    We first show that \Cref{alg:general} outputs all the minimal $(\Pi, w, k)$-sets of $U$.
    \revise{In \Cref{alg:general}, we use $\mathcal Q$ as a queue to traverse the directed graph $\mathcal G$ and $\mathcal O$ as the set of nodes that are already output (or equivalently traversed).}
    Let $S$ be a minimal $(\Pi, w, ck)$-set, which is given as input.
    Let $\mathcal O_{\mathcal Q}$ be the collection of minimal $\Pi$-sets generated by \Cref{alg:general} that are added to $\mathcal Q$.
    More formally, $\mathcal O_{\mathcal Q} = \inset{O \in \mathcal O}{w(O) \le w(S) + k}$.
    Suppose for contradiction that there is a minimal $(\Pi, w, k)$-set $Y$ that is not contained in $\mathcal O$.
    We choose a minimal $\Pi$-set $X$ and a minimal $(\Pi, w, k)$-set $Y$ in such a way that
    \revise{
        \begin{enumerate}
            \item there is a minimal $(\Pi, w, w(S))$-set $X' \in \mathcal O_{\mathcal Q}$ such that $X \subseteq X' \cup Y$ and
            \item $\size{X \cup Y}$ is minimized over all minimal $\Pi$-sets $X \in \mathcal O_{\mathcal Q}$ subject to the condition 1.
        \end{enumerate}
    }
    \revise{%
    Let us note that such a set $X$ can be chosen as $S$ itself satisfies the condition 1.
    }
    Let $x \in X \setminus Y$.
    By Lemma~\ref{lem:input-restricted}, $\mathcal S_{\mathcal A}$ contains $R \subseteq U$ such that $(X \setminus \{x\}) \cup R$ is a $\Pi$-set and $\comp{(X \setminus \{x\}) \cup R} \cup Y \subset  X \cup Y$.
    \revise{Let $Z = \comp{(X\setminus \{x\}) \cup R}$.}
    Then, we have $\size{X \cup Y} > \size{Z \cup Y}$.
    \revise{Moreover, as $X \subseteq X' \cup Y$ and $(X\setminus \set{x}) \cup R \subseteq X \cup Y$, we have
    \begin{align*}
        Z = \comp{(X \setminus \set{x}) \cup R} \subseteq (X \setminus \set{x}) \cup R \cup Y \subseteq X\cup Y \subseteq X' \cup Y,
    \end{align*}}%
    which contradicts the choices of $X$ and $Y$ \revise{since $w(Z) \le w(X' \cup Y) \le w(S) + k$.}
    Since the approximation factor of $\mathcal A$ is~$d$ and every minimal $\Pi$-set $X$ added into $\mathcal Q$ has weight at most $(c+1)k$, every minimal $\Pi$-set $Z$ generated by \Cref{alg:general} has weight at most
    \begin{linenomath}
        \begin{align}\label{eq:weight-bound}
            w(Z) = w(\comp{(X \setminus \{x\})} \cup R) \le w(X) + w(R) \le (c + d + 1)k.     
        \end{align}
    \end{linenomath}
    Hence, the approximation factor of \Cref{alg:general} is $c + d + 1$.

    As for the time complexity, suppose first that $\mathcal S_{\mathcal A}$ can be computed in output polynomial time, namely $O((|U| + |\mathcal S_{\mathcal A}|)^t)$ for some constant $t$ for each minimal $\Pi$-set $X$ and $x \in X$.
    By Lemma~\ref{lem:comp:uniqueness}, for $R \in \mathcal S_{\mathcal A}$, $\comp{(X \setminus \{x\}) \cup R}$ is distinct from $\comp{(X \setminus \{x\}) \cup R'}$ for any $R' \in \mathcal S_{\mathcal A}$ with $R' \neq R$.
    Thus, we have $|\mathcal S_{\mathcal A}| \le |\mathcal O|$.
    Therefore, the total running time is {upper-bounded} by
    \begin{linenomath}
    \[
    \sum_{X \in \mathcal O}\sum_{x \in X}O((|U|+|\mathcal S_{\mathcal A}|)^t) = O((|U| + |\mathcal O|)^{t + 1}),
    \]
    \end{linenomath}
    which is polynomial in $|U| + |\mathcal O|$.
    
    Suppose next that $\mathcal A$ runs in incremental polynomial time.
    Let $\mathcal O' \subseteq \mathcal O$ be a set of minimal $\Pi$-sets that have already been generated by the entire algorithm at some point in the execution.
    Now, consider the running time of the loop at line~\ref{alg:general:comp:neighbors}.
    Since $\mathcal Q$ contains sets in $\mathcal O'$, the loop repeats at most $|\mathcal O'|$ times without outputting a minimal $\Pi$-set $Z$ at line~\ref{alg:general:output:X}. 
    Moreover, for each $X \in \mathcal O'$ and $x \in X$, either $\comp{(X \cup \{x\}) \cup R} \notin \mathcal O'$ for some $R \in \mathcal S_{\mathcal A}$ or $\comp{(X \setminus \{x\}) \cup R} \in \mathcal O'$ for all $R \in \mathcal S_{\mathcal A}$.
    By Lemma~\ref{lem:comp:uniqueness}, for each $R \in \mathcal S_{\mathcal A}$, $\comp{(X \setminus \{x\}) \cup R}$ is a unique solution among $\mathcal S_{\mathcal A}$.
    This implies that we can either find a solution $\comp{(X \setminus \{x\} \cup R)} \notin \mathcal O'$ or conclude that no $R \in \mathcal S_{\mathcal A}$ with $\comp{(X \setminus \{x\}) \cup R} \notin \mathcal O'$ exists in time $O(\sum_{1 \le i \le |\mathcal O'|} (i \cdot |U|)^t) = O(|U|^t|\mathcal O'|^{t+1})$, where $t$ is a constant depending on the running time of $\mathcal A$.
    Therefore, the delay of the running time after generating $\mathcal O'$ is {upper-bounded} by
    \begin{linenomath}
    \[
        \sum_{X \in \mathcal O'}\sum_{x \in X}O(|U|^t\cdot |\mathcal O'|^{t+1}) = O(|U|^{t+1}\cdot|\mathcal O'|^{t+2}),
    \]
    \end{linenomath}
    which yields an incremental polynomial time bound.
\end{proof}


When we can solve \IRP{\Pi}{w}{k} in $|U|^{O(1)}$ time, we can simultaneously improve the running time and approximation factors of {\sc Minimal $(\Pi, w, k)$-Sets Enumeration Problem}.

\begin{toappendix}
\begin{algorithm}[t]
    \SetKwInput{KwInput}{Input}
    \SetKwInput{KwOutput}{Output}
    \SetAlgoLined
    \KwInput{A minimal $\Pi$-set $S$ of $U$ with $w(S) \le ck$ for some $c \ge 1$}
    Output $S$ and add $S$ to $\mathcal Q$ and $\mathcal O$\;
    \While{$\mathcal{Q}$ is not empty}{
        Let $X$ be a minimal $\Pi$-set in $\mathcal Q$\;
        Delete $X$ from $\mathcal Q$, add $X$ to $\mathcal O$, and output $X$\; 
        \ForEach{$x \in X$}{ 
            Compute $\mathcal S(U, \Pi, w, k, X, x)$\;
            \ForEach{$R \in \mathcal S(U, \Pi, w, k, X, x)$}{
                $Z \gets \comp{(X \setminus \set{x}) \cup R}$\; 
                \If{$Z \notin \mathcal O$}{
                    \If{$w(Z) < w(S) + k$}{
                    Add $Z$ to $\mathcal Q$\; 
                    } 
                }
            }
        }
    }
    \caption{A polynomial-delay $(c+1)$-approximate enumeration algorithm for {\sc Minimal $(\Pi, w, k)$-Sets Enumeration Problem}, where $\Pi$ is a CKS property. }
    \label{alg:general:cks}
\end{algorithm}
\end{toappendix}

\begin{theoremrep}\label{theo:apx:poly}
    Suppose that $\mathcal S$ can be enumerated in $\size{U}^{\order{1}}$ total time. 
    Then, one can solve {\sc Minimal $(\Pi, w, k)$-Sets Enumeration Problem} with approximation factor $c + 1$ in polynomial delay. 
\end{theoremrep}

\begin{proof}
    It is not hard to see that \Cref{alg:general:cks} outputs all minimal $(\Pi, w, k)$-sets by the same argument as in \Cref{theo:apx:output}.
    Since all steps inside {the} while loop at line~2 can be done in polynomial time, the delay is {upper-bounded} by a polynomial in $\size{U}$.
    Moreover, since we add only $(\Pi, w, (c+1)k)$-sets to $Q$, approximation factor of the algorithm is $c+1$.
\end{proof}

To apply \Cref{theo:apx:output,theo:apx:poly} to a specific monotone property $\Pi$, we need to develop a polynomial-time algorithm for solving \IRP{\Pi}{w}{k}. 
Cohen, Kimelfeld, and Sagiv~\cite{DBLP:journals/jcss/CohenKS08} and Cao~\cite{DBLP:journals/corr/abs-2004-09885} showed that 
{there are several graph properties $\Pi$ such that the \IRP{\Pi}{w}{k} can be solved exactly in polynomial time.
In what follows, we call such properties \emph{CKS properties}~\cite{DBLP:journals/corr/abs-2004-09885}.}
Cohen \etal and Cao showed that the following graph properties are CKS properties. 

\begin{proposition}[\cite{DBLP:journals/jcss/CohenKS08,DBLP:journals/corr/abs-2004-09885}]\label{prop:cks}
    Let $G = (V, E)$ be a graph. 
    Let $\mathcal C$ be one of the following classes of graphs; complete graphs, graphs with no edges, cluster graphs, complete bipartite graphs, complete $p$-partite graphs for any positive integer $p$, complete split graphs, split graphs, pseudo-split graphs, threshold graphs, and graphs of maximum degree $d$ for fixed $d$.
    Suppose that $\Pi$ is the property of being a vertex subset $S \subseteq V$ such that $G[V \setminus S]$ is in $\mathcal C$.
    Then, $\Pi$ is a CKS property.
\end{proposition}

{Note that if a graph belongs to a graph class $\mathcal C$ characterized by a set of forbidden subgraphs with at most constant $d$ vertices, then the problem of finding a vertex set whose removal results a graph contained in $\mathcal C$ to that of finding a hitting sets of size at most $k$ in a hypergraph of rank~$d$ (see \Cref{ssec:dHS} for more information).
Hence, }the problem of computing a minimum vertex set whose removal leaves a subgraph in a graph class listed in \Cref{prop:cks} can be approximated with a constant factor in polynomial time. 
Therefore, our approximate enumeration framework can be applied to those properties $\Pi$ and enumerate minimal $(\Pi, w, k)$-sets with constant approximation factors in polynomial delay.
Moreover, this framework can be used for more general weight functions, including monotone submodular functions.
Let us note that our frameworks can be used for more general weight functions.
A function $f: 2^U \to \mathbb Q_{> 0}$ is \emph{monotone} if $f(X) \le f(Y)$ for $X \subseteq Y \subseteq U$ and is \emph{subadditive} if $f(X) + f(Y) \ge f(X \cup Y)$ for $X, Y \subseteq U$.
One of the famous examples of subadditive functions is a submodular function. 
A function is \emph{submodular} if $f(X) + f(Y) \ge f(X \cup Y) + f(X \cap Y)$ for $X, Y \subseteq U$.
To bound the approximation factor of our enumeration algorithms for a monotone subadditive function $f$, we can replace function $w$ with $f$ in inequality~(\ref{eq:weight-bound}):
\begin{linenomath}
    \begin{align*}
        f(Z) = f(\comp{(X \setminus \{x\})} \cup R) \le f(X \cup R) \le f(X) + f(R) \le (c + d + 1)k,
    \end{align*}
\end{linenomath}
where the first and second inequalities are obtained from the monotonicity and subadditivity of $f$, respectively.
This gives a polynomial-delay $3$-approximate enumeration for minimal vertex covers with monotone submodular cost by combining \Cref{theo:apx:poly} with a polynomial-time $2$-approximation algorithm for the minimum vertex cover problem with submodular cost function due to~\cite{DBLP:conf/focs/IwataN09}.

\begin{toappendix}
\section{Applications of our framework}\label{appendix:application}
In this section, we give some concrete examples for efficient approximation enumeration algorithms. 

\subsection{Approximate Enumeration for Vertex Deletion Problems with Width Parameters}

In this subsection, we particularly work on \emph{vertex deletion properties on graphs}.
Let $G = (V, E)$ be a graph.
We say that a monotone property $\Pi$ over $V$ is a \emph{vertex deletion property} if for every minimal $\Pi$-set $X$ of $V$ and $R \subseteq X$, $X \setminus R$ is a minimal $\Pi$-set of $G[V \setminus R]$.
Many properties can be formulated as vertex deletion properties.
In particular, for a (possibly infinite) family of graphs $\mathcal F$, the property of being a vertex set whose removal breaks all induced subgraphs isomorphic to any graph in $\mathcal F$ is a vertex deletion property.

To solve hard problems on graphs, \emph{width parameters}, such as \emph{treewidth} and \emph{cliquewidth}, have become ubiquitous tools, and Courcelle's theorem~\cite{Courcelle:monadic:1990} is a cornerstone result to develop efficient algorithms on graphs when these parameters are small.
If the treewidth (resp. cliquewidth) of a graph is {upper-bounded} by a constant, the problem of finding a minimum weight vertex or edge set satisfying a property expressible by a formula in MSO$_2$ (resp. MSO$_1$) can be solved in linear time~\cite{Arborg:easy:1991,Courcelle:monadic:1990}.
There are numerous properties that are expressible in MSO$_1$ and MSO$_2$.
We refer to \cite{ParameterizedAlgorithms,Kreutzer:Algorithmic:2011} for the detailed definitions of these width parameters and a brief introduction to Courcelle's theorem.

For our purpose, we use the following lemmas.

\begin{lemmarep}\label{lem:tw}
    Let $G = (V, E)$ be a graph with constant treewidth and let $\Pi$ be a monotone property on $V$.
    Suppose that $\Pi$ can be expressed by a formula in MSO$_2$.
    Then, for $x \in V$, one can enumerate all the minimal $(\Pi, w, k)$-sets excluding $x$ in polynomial delay.
\end{lemmarep}
\begin{proof}
    Let $X \subseteq V$.
    Let $\phi(X)$ be a predicate that is true if and only if $X$ is a $\Pi$-set, which is expressed by a formula in MSO$_2$.
    As $\Pi$ is monotone, the minimality of $X$ can be expressed as:
    \begin{linenomath}
    \[
        \phi_m(X) := \phi(X) \land \forall v\in V.(v \in X \implies \neg\phi(X \setminus \{v\})).
    \]
    \end{linenomath}    
    Let $I \subseteq V$ and let $O \subseteq V \setminus I$.
    Then, the property that there is a minimal $\Pi$-set $X$ excluding $x$ such that $I \subseteq X$ and $O \cap X = \emptyset$ is expressible with the following formula:
    \begin{linenomath}
    {
        \begin{align*}
            \phi_m(x, I, O) := 
            & \exists X \subseteq V. (\phi_m(X) \land \forall v \in V.((v \in I \implies v \in X) \land \\
            & (v \in (O \cup \{x\}) \implies v \not\in X))).
        \end{align*}
    }
    \end{linenomath}        
    Now, we can enumerate all the minimal $(\Pi, w, k)$-sets of excluding $x$ in polynomial delay by the $K$-best enumeration technique~\cite{Lawler1972}.
    \begin{itemize}
        \item We first compute a minimum weight $(\Pi, w, k)$-set $S$ excluding $x$ by the optimization version of Courcelle's theorem with MSO${}_2$ formula $\phi_m(x, \emptyset, \emptyset)$.
        If $w(S) > k$, we do nothing. Otherwise, we add a tuple $(S, \emptyset, \emptyset)$ to a queue.
        \item Remove a tuple $(S', I, O)$ from the queue that minimizes $w(S')$ among those in the queue.
        Output $S'$ and do the following. We repeat this and the subsequent step unless the queue is empty.
        \item Let $v_1, v_2, \ldots, v_t$ be the vertices in {$S' \setminus I$}.
        For each {$1 \le i < t$}, we compute a minimal $(\Pi, w, k)$-set {$S_i$} including {$I \cup \{v_1, \ldots, v_i\}$} and excluding {$O \cup \{v_{i+1}, x\}$} with MSO${}_2$ formula {$\phi_m(x, I \cup \{v_1, \ldots, v_i\}, O \cup \{v_{i+1}\})$} that minimizes {$w(S_i)$}.
        If {$w(S_i) \le k$}, add {$(S_i, I \cup \set{v_1, \ldots, v_i}, O \cup \set{v_{i+1}})$} to the queue.
    \end{itemize}
    {
        From the above procedure, all outputs are minimal $(\Pi, w, k)$-set excluding~$x$ with the weight at most $k$.
        We denote by $\mathcal O$ the set of solutions that are output by the algorithm.
        We show that $\mathcal O$ contains all minimal $(\Pi, w, k)$-sets excluding~$x$ with the weight at most $k$. 
        Suppose otherwise.
        Let $T \notin \mathcal O$ be a minimal $(\Pi, w, k)$-set excluding $x$ with the weight at most $k$.
        For each solution $S \in \mathcal O$, we denote by $(S, I_S, O_S)$ the tuple that is removed from the queue when the algorithm outputs $S$.
        Let $\tilde S \in \mathcal O$ that satisfies $T \cap (O_{\tilde S} \cup \set{x}) = \emptyset$ and $I_{\tilde S} \subseteq T$, and maximizes $\size{T \cap I_{\tilde S}}$.
        We can choose such a set as $(S, \emptyset, \emptyset)$ is added to the queue.
        Let $v_1, \dots, v_t$ be the vertices in $\tilde S \setminus I_{\tilde S}$.
        As $\tilde S \neq T$ and $I_{\tilde S} \subseteq T$,
        there is a vertex $v_{i+1} \in \tilde S \setminus I_{\tilde S}$ such that $v_{i+1} \not\in T$ and $\set{v_1, \ldots, v_i} \subseteq T$.
        Let $\tilde S_i$ be a minimal $(\Pi, w, k)$-set with the minimum weight that satisfies $I_{\tilde S} \cup \set{v_1, \ldots, v_i} \subseteq \tilde S_i$ and $(O \cup \set{v_{i+1}, x}) \cap \tilde S_i = \emptyset$.
        Since $T$ satisfies the constraints, the weight of $\tilde S_i$ is at most the weight of $T$.
        Hence, $\tilde S_i$ is contained in $\mathcal O$ and
        it contradicts the fact that $\tilde S$ is chosen to maximize $\size{T \cap I_{\tilde S}}$ subject to
        $T \cap (O_{\tilde S}\cup\set{x}) = \emptyset$ \revise{and} $I_{\tilde S} \subseteq T$.
        This ensures that the algorithm enumerates all minimal $(\Pi, w, k)$-set excluding $x$ with the weight at most $k$.
        Finally, each step can be done in polynomial time and hence the delay is upper-bounded by a polynomial in $\size{V}$.
    }
\end{proof}

We also give an enumeration algorithm on bounded cliquewidth graphs.
{
    The proof is almost the same as the previous one.
    Similar to the previous proof, we can find a minimal $(\Pi, w, k)$-set expressed by a formula in MSO${}_1$ using the above technique.
    Thus, we omit the proof of the following lemma.
}
\begin{lemma}\label{lem:cw}
    Let $G = (V, E)$ be a graph with constant cliquewidth and let $\Pi$ be a monotone property over $V$.
    Suppose that $\Pi$ can be expressed by a formula in MSO$_1$.
    Then, for $x \in V$, one can enumerate all the minimal $(\Pi, w, k)$-sets excluding $x$ in polynomial delay.
\end{lemma}

Now, we show the main result of this subsection.

\begin{theoremrep}\label{thm:enum-width}
    Let $G = (V, E)$ be a graph and let $\Pi$ be a monotone vertex deletion property over $V$ that is expressible by a formula in MSO$_2$ (resp. MSO$_1$).
    Suppose that for every $\Pi$-set $X$, the graph obtained from $G$ by deleting $X$ has constant treewidth (resp. constant cliquewidth).
    Given a $\Pi$-set $S$ of weight at most $ck$ for some $c > 0$, one can enumerate all the minimal $(\Pi, w, k)$-sets with approximation factor $c + 2$ in incremental polynomial time.
\end{theoremrep}
\begin{proof}
    Let $X \subseteq V$ be a minimal $\Pi$-set.
    Suppose that for every $\Pi$-set $X$, the treewidth (resp. cliquewidth) of $G[V \setminus X]$ is at most a constant.
    Since adding a vertex increases its treewidth by at most one (resp. cliquewidth to at most twice plus one \cite{Lozin:Band:2004}), $G[V \setminus (X \setminus \{x\})]$ also has a constant treewidth (resp. cliquewidth).
    For $x \in X$ and $R \subseteq V \setminus (X \setminus \{x\})$, let $\phi(R, x)$ be the predicate that is true if and only if $R$ is a minimal $\Pi$-set of $G[V \setminus (X \setminus \{x\})]$ excluding $x$, which is expressed by a formula in MSO$_2$ (resp. MSO$_1$).
    By Lemmas~\ref{lem:tw} and~\ref{lem:cw}, there is an algorithm $\mathcal A$ that compute $\mathcal S_{\mathcal A}(V, \Pi, w, k, X, x)$ exactly in polynomial delay, and hence by Theorem~\ref{theo:apx:output}, we can enumerate all the minimal $(\Pi, w, k)$-sets in incremental polynomial time with approximation factor $c + 2$.
\end{proof}

Let $\mathcal F$ be a finite set of graphs that contains at least one planar graph.
It is known that every graph that does not contain a fixed planar graph $H$ as a minor has a constant treewidth~\cite{Robertson:GM5:1986}.
The property $\Pi_{\mathcal F}$ of being a vertex set that hits all the minors of graphs in $\mathcal F$ can be expressed in a MSO$_2$ formula.
Moreover, there is an $O(1)$-approximate algorithm for finding a smallest {\em cardinality} $\Pi_{\mathcal F}$-set \cite{Fomin:Planar:2012} if $\mathcal F$ contains at least one planar graph.
Note that for weighted graphs, the best known approximation factor for this problem is polylogarithmic~\cite{Agrawal:Polylogarithmic:2018}.
These facts yield the following corollary.

\begin{corollary}\label{cor:enum-treewidth}
    Let $G = (V, E)$ be a graph.
    Let $\mathcal F$ be a set of graphs that contains at least one planar graph and
    let $\Pi_{\mathcal F}$ be the vertex deletion property corresponding to every set $S$ whose removal from $G$ leaves a graph {that} does not contain any graph in $\mathcal F$ as a minor.
    Then, one can enumerate all the minimal $(\Pi, k)$-sets with a constant approximation factor in incremental polynomial time.
\end{corollary}

By Corollary~\ref{cor:enum-treewidth}, we have incremental polynomial-time approximate enumeration algorithms for minimal feedback vertex sets, minimal caterpillar forest vertex deletion sets, minimal star forest vertex deletion sets, minimal outerplanar vertex deletion sets, and minimal treewidth-$\eta$ vertex deletion sets for fixed integer $\eta$.
    We note that the approximation algorithm due to~\cite{Fomin:Planar:2012} is randomized and the expected approximation factor is {upper-bounded} by a constant.
Also note that, for specific problems, such as the minimum feedback vertex set problem, the minimum caterpillar forest vertex set problem (also known as the pathwidth-one vertex set deletion problem), and the minimum star forest vertex deletion set problem, they admit deterministic constant factor approximation algorithms~\cite{Becker:Optimization:1996,Philip:Quartic:2010,demaine_et_al:LIPIcs:2019:11158} In \cite{Philip:Quartic:2010}, they did not explicitly mention the approximability of the minimum caterpillar forest vertex deletion set problem, while we can easily obtain a deterministic $7$-approximation algorithm from their argument. {
A graph is a forest of caterpillars if and only if it has no cycles or $2$-claws as a subgraph~(Lemma~1 in~\cite{Philip:Quartic:2010}), where a graph is a $2$-claw if it is obtained from a claw $K_{1, 3}$ by attaching a pendant vertex to each leaf.
Moreover, if a graph has no cycles of length~$3$, a cycle of length~$4$, or a $2$-claw as a subgraph, each connected component is a tree, or a cycle with ``hairs'' (Lemma~3 in \cite{Philip:Quartic:2010}).
Using these facts, we have a polynomial-time approximation algorithm. We repeatedly find a subgraph isomorphic to a cycle of length at most $4$ or a $2$-claw and then delete all the vertices in the subgraph.
Then, we compute a minimum caterpillar forest vertex deletion set in the remaining graph.
Since each remaining component is either a tree or a cycle with hairs, we can solve it in polynomial time.
Moreover, the approximation ratio of the above algorithm is at most $7$, as each subgraph deleted in the first step contains at most seven vertices.
}


As for dense graph classes, some graph classes have constant cliquewidth: Block graphs, cographs, bipartite chain graphs, and trivially perfect graphs are of bounded cliquewidth.
Since the classes of cographs, bipartite chain graphs, and trivially perfect graphs can be respectively characterized by finite sets of forbidden induced subgraphs,
there are $O(1)$-approximation algorithms for finding a minimum weight set $S \subseteq V$ such that $G[V \setminus S]$ is a cograph, a bipartite chain graph, or a trivially perfect graph.
Although no characterization by a finite set of forbidden induced subgraphs for block graphs is known, {there is an approximation algorithm for the minimum block graph vertex deletion problem on unweighted graphs with the approximation factor $4$. See Theorem~6 in \cite{Agrawal:faster:2016}}. Moreover, the property of being a block graph is expressible in MSO$_1$ by using the fact that a graph is a block graph if and only if every biconnected component induces a clique.
Therefore, as another consequence of Theorem~\ref{thm:enum-width}, we have the following corollary.
\begin{corollary}\label{cor:enum-cliquewidth}
    Let $G = (V, E)$ be a graph.
    Let $\mathcal C$ be one of the following classes of graphs; cographs, bipartite chain graphs, and trivially perfect graphs.
    Let $\Pi$ be the property of being a set of vertices whose removal leaves a graph in $\mathcal C$.
    Then, one can enumerate all the minimal $(\Pi, w, k)$-sets with constant approximation factors in incremental polynomial time.
    For block graphs, one can enumerate all the minimal $(\Pi, k)$-sets with approximation factor $6$ in incremental polynomial time {since there is an approximation algorithm for the minimum block vertex deletion problem with the approximation factor $4$~\cite{Agrawal:faster:2016}}.
\end{corollary}

Since the property of being a set of vertices whose removal leaves a distance hereditary graph is monotone and every distance hereditary graph has cliquewidth at most five, we can apply our framework to this property.
However, there is no known polynomial-time constant factor approximation algorithm for finding a minimum vertex set of this property.
The current best approximation factor is polylogarithmic in the input size due to Agrawal et al.~\cite{Agrawal:Polylogarithmic:2018}.

\subsection{\texorpdfstring{$d$}{d}-Hitting Set}\label{ssec:dHS}
Let $\mathcal H = (V, \mathcal E)$ be a hypergraph. A \emph{hitting set} of $\mathcal H$ is a subset $S$ of $V$ such that $S \cap e \neq \emptyset$ for every $e \in \mathcal E$.
The existence of an output-polynomial time enumeration algorithm for minimal hitting sets is a long-standing open problem in this field and the best known algorithm runs in quasi-polynomial time in the size of input and output~\cite{Fredman:complexity:1996}.
There are several studies for developing efficient enumeration algorithms for special hypergraphs.
In particular, if every hyperedge contains at most $d$ vertices for some fixed constant $d \ge 2$, 
there is an incremental polynomial time enumeration algorithm~\cite{DBLP:conf/latin/BorosEGK04}.
We say that such a hypergraph has \emph{rank} at most $d$ and call a hitting set of rank-$d$ hypergraphs a \emph{$d$-hitting set}.

We show that minimal $d$-hitting sets of hypergraphs with weight at most~$k$ can be enumerated in incremental-polynomial time with approximation factor $\frac{(d + 4)(d - 1)}{2}$.
It is known that the problem of computing a minimum weight $d$-hitting set of $\mathcal H$ has a polynomial-time $d$-approximation algorithm. 
We use this algorithm to compute a seed set.

\begin{theoremrep}\label{theo:dhit}
    Let $\mathcal H = (V, \mathcal E)$ be a hypergraph with rank at most $d$.
    Then, there is an incremental-polynomial $\frac{(d + 4)(d - 1)}{2}$-approximate enumeration algorithm for enumerating minimal $d$-hitting sets of $\mathcal H$.
\end{theoremrep}
\begin{proof}
    Let $k$ be positive. We show, by induction on $d$, that there is an approximate enumeration algorithm for minimal $d$-hitting sets of $\mathcal H$ of weight at most $k$.
    If $d = 2$, then the problem is equivalent to the vertex cover case, and hence, by Proposition~\ref{prop:cks}, we have a $3$-approximate polynomial delay enumeration algorithm.
    
   Now, we assume that we can enumerate in incremental-polynomial time all the minimal $(d-1)$-hitting sets of weight at most $k$ with the approximation factor $\frac{(d + 3)(d - 2)}{2}$ for every hypergraph having rank at most $d - 1$.
   We consider an enumeration algorithm for minimal $d$-hitting sets of $\mathcal H$.
   Let $\Pi_{\tt HS}$ be the property of being a hitting set of a hypergraph $\mathcal H$. 
   Observe that, for a minimal hitting set $X \subseteq V$ of $\mathcal H$ and $x \in X$, {\sc Input-Restricted Minimal $(\Pi_{\tt HS}, w, k)$-Sets Problem} for $\mathcal H$ can be seen as the problem of enumerating minimal $(d-1)$-hitting sets in a hypergraph $\mathcal H'$ obtained from $\mathcal H$ by deleting all the hyperedges intersecting with $X \setminus \set{x}$ and removing $x$ from each remaining hyperedge.
   By the induction hypothesis, we can enumerate minimal $(d-1)$-hitting sets of $\mathcal H'$ with approximation factor $\frac{(d + 3)(d-2)}{2}$ in incremental-polynomial time.
   Since there is a polynomial-time $d$-approximation algorithm for the minimum $d$-hitting set problem, by \Cref{theo:apx:output}, we can enumerate all the minimal $d$-hitting sets of $\mathcal H$ in incremental-polynomial time with approximation factor $d + \frac{(d+3)(d-2)}{2} + 1 = \frac{(d+4)(d-1)}{2}$.
\end{proof}

\subsection{Star Forest Edge Deletion}\label{subsec:sfed}

\newcommand{\SFED}{{\tt SFED}}
\newcommand{\psfvd}{\Pi_{\tt SFVD}}

A graph is called a \emph{star forest} if each component is a star graph.
A \emph{star graph} is a graph isomorphic to $K_{1, t}$, a complete bipartite graph of $t + 1$ vertices
{with partition $A$ and $B$ such that $\size{A} = 1$ and $\size{B} = t$.}
The \emph{center} of a star is a vertex of degree greater than one. 
Note that we regard an isolated vertex or a complete graph of two vertices as a star. 
A vertex $v$ is called a \emph{leaf} if degree of $v$ is equal to one.
Let $G = (V, E)$ be a graph.
A vertex set $S \subseteq V$ (resp. an edge set $F \subseteq E$) is called a \emph{star forest vertex deletion set} (resp. \emph{star forest edge deletion set}) of $G$ if $G[V \setminus S]$ (resp. $G - F$) is a star forest.
Let $\psfvd$ and $\Pi_\SFED$ be the properties of being a star forest vertex and edge deletion set of $G$, respectively.
It is easy to see that both properties are monotone over $V$ and $E$, respectively.

\begin{lemmarep}\label{lem:cks-sfed}
    $\Pi_\SFED$ is a CKS property.
\end{lemmarep}

\begin{proof}
    Let $G = (V, E)$ be a graph, $X \subseteq E$ be a minimal star forest edge deletion set of $G$, and $e = \{u, v\} \in X$. 
    Let $\mathcal S$ be a collection of a set of edges $R$ such that $(X \setminus \set{e}) \cup R$ is star forest edge deletion. 
    Since $X$ is a minimal star forest deletion, $G - (X\setminus\set{e})$ contains exactly one component $C$ that is not a star. 
    Hence, $R$ only contains edges in $C$. 

    We consider two cases:
    (1)~exactly one of $d_C(u)$ or $d_C(v)$ is more than one, or
    (2)~both $d_C(u)$ and $d_C(v)$ are more than one.    
    In Case~(1), we may assume $d_C(u) > 1$.
    Since $C$ is not a star and $C - \set{e}$ is a star, 
    $C$ contains~$w$ with $d_C(w) > 1$ that is adjacent to $u$ in $C$. 
    This implies that $u$ is a leaf in $G - X$. 
    There are two minimal star forest edge deletion sets $R$ of $C$ without containing $e$: $\set{\set{u, w}}$ and $\Gamma_C(w)\setminus \set{\set{u, w}}$.
    Suppose that both $d_C(u)$ and $d_C(v)$ are more than one, namely Case~(2).
    As $R$ does not contain $e$, at least one of $u$ and $v$, say $u$, is a leaf in $C - R$.
    Thus, $R$ contains all edges but $e$ incident to $u$.
    If $v$ is also a leaf in $C - R$ (i.e., $e$ is a component in $C - R$), then $R$ contains all edges but $e$ incident to $v$ as above.
    Otherwise, $v$ is the center vertex of degree at least two in the star component $C'$ of $C - R$.
    Then, every neighbor $w$ of $v$ in $C'$ must be a leaf.
    Therefore, $R$ contains all edges incident to $w \in N_C(v)$ except one edge $\{v, w\}$.
    By guessing whether $u$ and~$v$ are leaves or centers in $C - R$, we can uniquely determine minimal star forest edge deletion set $R$ of $C$.
    
    In all cases, there are only a constant number of possibilities of minimal star forest edge deletion set $R$ of $C$.
    Hence, $\mathcal S$ contains a constant number of sets, which can be computed in polynomial time.
\end{proof}

Since there is a polynomial-time $3$-approximation algorithm for finding a minimum star forest edge deletion set~\cite{demaine_et_al:LIPIcs:2019:11158}, by Theorem~\ref{theo:apx:poly}, we have the following theorem.

\begin{theorem}\label{thm:approx-enum-sfd}
    There is an approximate enumeration algorithm for minimal star forest edge deletion sets with approximation factor $4$ that runs in polynomial delay.
\end{theorem}

Contrary to this, we cannot directly apply our framework to the vertex counterpart.
We should remark that this does not imply there is no polynomial-delay enumeration algorithm with a constant approximation factor. 

\begin{propositionrep}\label{prop:sfvd-ncks}
    $\psfvd$ is not a CKS property.
\end{propositionrep}
\begin{proof}
    From a star graph $K_{1, 2n}$ with leaves $\{v_1, v_2, \ldots, v_{2n}\}$ for some integer~$n$ and the center $r$, we construct a graph $G = (V, E)$ by adding an edge between~$v_{2i - 1}$ and $v_{2i}$ for each $1 \le i \le n$. 
    Let $X = \{r\}$.
    Then, $X$ is a minimal star forest vertex deletion set of $G$.
    It is easy to verify that there are exponentially many minimal star forest vertex deletion sets $Y$ of $G[V]$ that excludes $r$ since for each $1 \le i \le n$, $Y$ contains either $v_{2i - 1}$ or $v_{2i}$, which implies that {there are} $2^n$ possible combinations {for} $Y$. 
\end{proof}

\subsection{Dominating Set in Bounded Degree Graphs}

Let $G = (V, E)$ be a graph.
We say that $D \subseteq V$ is a \emph{dominating set} of~$G$ if every vertex in $V$ is either contained in $D$ or adjacent to a vertex in $D$.
Let $\Pi_{\tt Dom}$ be the property of being a dominating set of $G$.
Although it is not hard to see that $\Pi_{\tt Dom}$ is not a CKS property,
we can prove that it does have a CKS property when the maximum degree of $G$ is {upper-bounded} by a constant.

\begin{lemmarep}\label{lem:cks-dom}
    $\Pi_{\tt Dom}$ is a CKS property, provided every vertex of $G$ has degree at most some fixed constant $\Delta$.
\end{lemmarep}

\begin{proof}
    Let $X$ be a minimal dominating set of $G$.
    Since the degree of every vertex of $G$ is at most $\Delta$, there are a constant number of non-dominated vertices $U \subseteq N(x)$ by $X \setminus \set{x}$.
    Moreover, since such a vertex in $U$ can be dominated by its neighbor, there are a constant number of possibilities of minimal vertex addition. 
    Therefore, $\mathcal S$ can be computed in polynomial time.
\end{proof}

It is well known that the minimum dominating set problem admits a polynomial-time $O(\log \Delta)$-approximation algorithm on graphs of maximum degree $\Delta$~\cite{Chvatal:Greedy:1979}. This yields the following approximate enumeration.
\begin{theorem}\label{thm:approx-enum-dom}
    There is an approximate enumeration algorithm for minimal dominating sets on bounded degree graphs with approximation factor $O(\log \Delta)$ runs in polynomial delay, where $\Delta$ is the maximum degree of input graphs.
\end{theorem}

\end{toappendix}

\section{Approximate enumeration beyond our frameworks}\label{sec:eds}
In this section, as another algorithmic contribution, we propose polynomial-delay constant factor approximate enumeration algorithms for minimal edge dominating sets and minimal Steiner subgraphs, for which \revise{it} seems to be difficult to apply the previous frameworks directly. 

\newcommand{\rom}[1]{\uppercase\expandafter{\romannumeral #1\relax}}

 \newcommand{\II}{\rom{2}\xspace}
 \newcommand{\III}{\rom{3}\xspace}
 \newcommand{\IV}{\rom{4}\xspace}

\subsection{Minimal edge dominating sets}
Let $G = (V, E)$ be a graph. Let $n = |V|$ and $m = |E|$. 
A set $D \subseteq E$ of edges is an \emph{edge dominating set} of $G$ if every edge $e$ in $E$ either belongs to $D$ {or has a common end vertex with some $f$ in $D$}.
Let $w\colon E \to \mathbb Q_{> 0}$ be a weight function. 
In the latter case, we say that~$e$ is \emph{dominated by $D$} or $D$ \emph{dominates~$e$}.
Let $X \subseteq E$ be an edge dominating set of $G$.
When an edge~$e \in E$ is dominated by $X$, but is not dominated by $X \setminus \set{x}$ for some $x \in X$, $e$ is called the \emph{private edge of $x$}. 
Moreover, we say that $x$ has the private edge $e$. 
Note that $x$ may have itself as the private edge.
It is known that $X$ is a minimal edge dominating set of $G$ if and only if any edge in $X$ has at least one private edge. 
Let $\Pi_{\tt EDS}$ be the property of being an edge dominating set of a graph $G = (V, E)$.  
{
    An important observation is that a minimal edge dominating set is a star forest.
    This observation is a key to our approximate enumeration algorithm.
}

Kant\'e \etal developed a polynomial-delay and polynomial-space enumeration algorithm for minimal edge dominating sets~\cite{DBLP:conf/wads/KanteLMNU15}. 
Their algorithm is based on the reverse search technique~\cite{Avis:Fukuda:DAM:1996}.
It would be highly nontrivial to extend their algorithm to our problem setting.
Moreover, unfortunately, $\Pi_{\tt EDS}$ is {not a CKS property:} 
The input-restricted problem of the instance in \Cref{fig:EDS-nonCKS} has the exponential number of solutions for $n$.
\begin{figure}[h]
    \centering
    \includegraphics[width=0.8\textwidth]{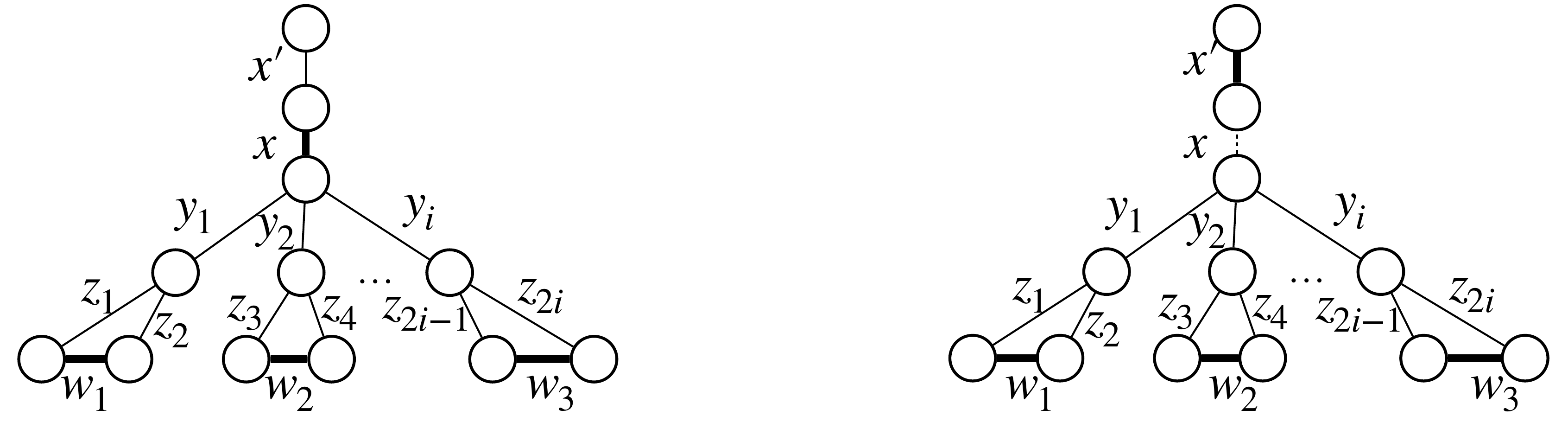}
    \caption{Example of an intractable case of the input-restricted problem.
    Thick lines indicate edges in $X$. There are at least $2^i$ minimal sets $Y$ such that $(X \setminus \{x\}) \cup Y$ is an edge dominating set of $G$ since we can take independently $z_{2j-1}$ or $z_{2j}$ for each $1 \le j \le i$. 
    Note that $(X \setminus \set{x}) \cup Y$ needs not to be minimal. }
    \label{fig:EDS-nonCKS}
\end{figure}
This implies that we cannot directly apply the framework in \Cref{sec:framework}. 

In this subsection, we show that it is still possible to approximately enumerate all minimal edge dominating sets in polynomial delay with a constant approximation factor. 
To make our idea clear, we first describe an algorithm for enumerating all minimal edge dominating sets of $G$ without weight constraints and postpone proving an approximation guarantee to the {end} of this subsection. 

The algorithm is also based on the supergraph technique. 
Thus, we need to build a directed graph $\mathcal G = (\mathcal V, \mathcal E)$, where $\mathcal V$ is the set of minimal edge dominating sets of $G$ and $\mathcal E$ is the set of arcs based on two types of neighborhood. The precise definition of neighbors will be given later. 
Intuitively speaking, to construct a directed path from $X$ to $Y$, we repeatedly exchange edges in $X \setminus Y$ with edges outside $X$. 
Let $x = \set{u, v} \in X \setminus Y$.
Note that $Y$ contains at least one edge $e$ incident to $u$ or $v$ as otherwise $x$ is not dominated by $Y$.
If~$Y$ contains two edges $e \in \Gamma(u)$ and $f \in \Gamma(v)$, then the type-I neighbor of~$X$ with respect to $(x, e, f)$ is ``closer'' than $X$ to $Y$.
On the other hand, if $Y$ does not contain any edges in $\Gamma(v)$, we cannot find a type-I neighbor which is ``closer'' to $Y$.
Then we dare to choose a ``detour'' by taking a type-\II neighbor of $X$ that covers $\Gamma(v)$ from the ``outside'' of $\Gamma(v)$. 
From this type-\II neighbor, we can construct a directed path to $Y$ by tracing a type-I neighbor of each minimal edge dominating set on the path.

\newcommand{\we}[1]{W_{u,x,e}(#1)}

To complete the description of $\mathcal G$, we define two types of neighbors of a minimal edge dominating set $X$ defined as follows:
\begin{itemize}
    \item For any edge $x = \set{u, v}$ in $X$, 
        let $e$ and $f$ be edges such that $e \in \Gamma(u) \setminus \set{x}$ and $f \in \Gamma(v) \setminus \set{x}$. 
        Note that $e$ or $f$ might belong to $X$. 
        We say that $Z_1 = \comp{(X \setminus \set{x}) \cup \set{e, f}}$ is the \emph{type-I neighbor of $X$ with respect to $(x, e, f)$}.
        When either $\Gamma(u) \setminus \{x\}$ or $\Gamma(v) \setminus \{x\}$ is empty, say $\Gamma(v) = \{x\}$, then we define $Z_1 = \comp{(X \setminus \set{x}) \cup \set{e}}$, and we also call it the type-I neighbor of $X$ with respect to $(x, e)$.
    \item For any edge $x = \set{u, v}$ in $X$,
        let $e \neq x$ be an edge that shares an end vertex, say $u$, with $x$.
        Define $\we{X}$ as an arbitrary minimal set of edges such that $(X \setminus \set{x}) \cup (\we{X} \cup \set{e})$ is an edge dominating set of $G$ with $\we{X} \cap \Gamma(v) = \emptyset$.
        Note that $\we{X}$ may not be appropriately defined since $\Gamma(v)$ may contain an edge that cannot be dominated by any edge in $E \setminus \Gamma(v)$.
        If $\we{X}$ is well-defined, we say that $Z_2 = \comp{(X \setminus \set{x}) \cup (\we{X} \cup \set{e})}$ is the \emph{type-\II neighbor of $X$ with respect to $(u, x, e)$}. 
\end{itemize}

Note that $\we{X}$ is a minimal set of edges that excludes $x$ and dominates edges of $\Gamma(v) \setminus \set{x}$ not dominated by $(X \setminus \set{x}) \cup \set{e}$.
Thus, every edge in $\we{X}$ has {a} private edge in $\Gamma(v)$.
This also implies $\we{X}$ contains at most $|\Gamma(v)| \le \Delta$ edges.
This observation is useful to analyze the delay of the algorithm in \Cref{theo:meds:poly}. 
We can easily see a type-I and type-\II neighbor of $X$ are always minimal edge dominating sets of $G$. 
\revise{Recall} that for an edge dominating set $X$, $\mu(X)$ is arbitrary minimal edge dominating set that is contained in $X$.

We first show that $\mathcal G$ is strongly connected.
Let $X$ and $Y$ be two distinct minimal edge dominating sets of $G$.
We wish to prove that there is a neighbor $Z$ of $X$ such that $|Z \cup Y| < |X \cup Y|$, that is, $Z$ is ``closer'' to $Y$ than $X$ in $\mathcal G$.
However, $X$ may not have such a neighbor. 
To prove the strong connectivity of $\mathcal G$, we show that if $X \neq Y$, then there always exists a desirable set $Z$ with $|Z \cup Y| < |X \cup Y|$ such that $\mathcal G$ has a directed path from $X$ to $Z$.

\begin{lemmarep}\label{lem:eds:typeI:has:good:neighbor}
    Let $X$ and $Y$ be distinct two minimal edge dominating sets of~$G$,
    {
        $x = \set{u, v}$ be an edge in $X \setminus Y$, and
        $e$ be an edge in $Y$ that incident to $u$. 
    }
    If (a) $Y$ contains an edge in $\Gamma(v)$ or (b) $X$ contains an edge in $\Gamma(v) \setminus \set{x}$, 
    then $X$ has a type-I neighbor $Z$ satisfying $\size{Z \cup Y} < \size{X \cup Y}$. 
\end{lemmarep}

\begin{proof}
    If (a) holds, that is, $Y$ contains an edge $f$ in $\Gamma(v)$,  
    then the type-I neighbor $Z$ of $X$ with respect to $(x, e, f)$ satisfies 
    $|Z \cup Y| \le |(X \setminus \{x\}) \cup \set{e, f} \cup Y| < |X \cup Y|$.
    
    In what follows, we suppose that (a) does not hold but (b) holds.  
    If $\Gamma(v)$ contains exactly one edge $x$, then the type-I neighbor $Z$ of $X$ with respect to $(x, e)$ satisfies $\size{Z \cup Y} \le \size{(X \setminus \set{x}) \cup \{e\} \cup Y} < \size{X \cup Y}$. 
    Thus, in what follows, we assume that $\Gamma(v)$ contains more than one edges. 
    Note that $\size{\Gamma(u)}$ is also more than one since $\Gamma(u)$ has at least two edges $e$ and $x$. 
    
    Suppose that $X$ contains edges in $\Gamma(v) \setminus \set{x}$. 
    As $Y$ has no edges incident to $v$ and $x \notin Y$, $Y$ contains at least one edge $e$ incident to $u$. 
    By the assumption that $X$ has an edge $f$ in $\Gamma(v) \setminus \set{x}$, $x$ has no private edges incident to $v$ at $X$, which implies that every private edge of $x$ at $X$ is incident to $u$. 
    Hence, $Z' = (X \setminus \set{x}) \cup \set{e, f}$ is an edge dominating set of $G$.
    Moreover, $\size{X \cup Y} > \size{Z' \cup Y} \ge \size{\comp{Z'} \cup Y}$ as $x \notin Y$ and $e \in Y$
    {since $Z = \comp{Z'}$ is the type-I neighbor of $X$ with respect to $(x, e, f)$. }
\end{proof}

We assume the condition of~\Cref{lem:eds:typeI:has:good:neighbor} does not hold, 
that is, $Y \cap \Gamma(v) = \emptyset$ and $X \cap \Gamma(v) = \set{x}$.
By this assumption, $Y$ contains at least one edge $e \in \Gamma(u)$ with $e \neq x$. 
Moreover, we assume that $\Gamma(v) \setminus \{x\}$ contains at least one edge as otherwise the type-I neighbor $Z = \comp{(X \setminus \set{x}) \cup \set{e}}$ of $X$ satisfies $|Z \cup Y| < |X \cup Y|$ and hence we are done.
From these assumptions, at $Y$, it follows that 
the edges in $\Gamma(v) \setminus \set{x}$ are dominated by edges not incident to $v$. 
Recall that $\we{X}$ is defined to be an arbitrary minimal set of edges such that $(X \setminus \set{x}) \cup (\we{X} \cup \set{e})$ is an edge dominating set with $\we{X} \cap \Gamma(v) = \emptyset$.
Such a set $\we{X}$ is well-defined since $Y$ contains edges not in $\Gamma(v)$ that {dominate} those in $\Gamma(v) \setminus \set{x}$. 
This observation yields the following lemma.

\begin{lemmarep}\label{lem:path}
    Let $X$ and $Y$ be distinct minimal edge dominating sets of $G$.
    Let $x \in X \setminus Y$ with $x = \set{u, v}$.
    Suppose that $Y \cap \Gamma(v) = \emptyset$, $X \cap \Gamma(v) = \{x\}$, and $Y$ contains an edge $e \neq x$ incident to $u$.
    Then, $X$ has a type-\II neighbor $Z^*_0$ such that there is a directed path from~$Z^*_0$ to a minimal edge dominating set $Z$ with $|Z \cup Y| < |X \cup Y|$. 
\end{lemmarep}
\begin{proof}
    \begin{figure}
        \centering
        \includegraphics[width=0.8\textwidth]{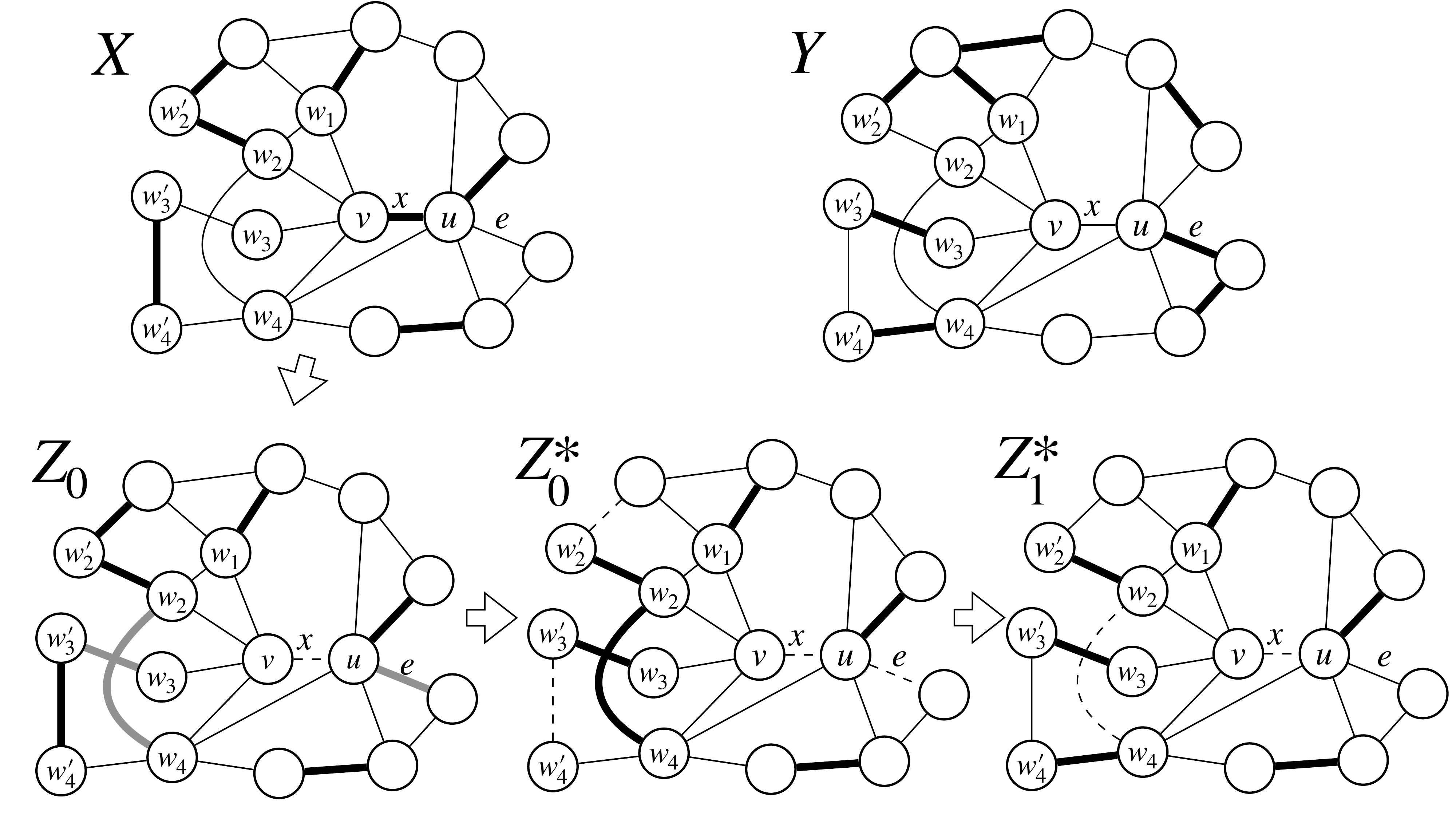}
        \caption{A running example of \Cref{lem:path}.
        {Grey thick lines mean edges added to a current edge dominating set.
        Dotted lines mean edges removed from a current edge dominating set.
        }
        We define $Z_0 = (X\setminus \set{x}) \cup \set{\set{w_3, w'_3}, \set{w_2, w_4}, e}$, where $\we{X} = \set{\set{w_3, w'_3}, \set{w_2, w_4}}$ and $Z^*_0 = \comp{Z_0}$.
        Edges $\we{X}$ dominate undominated edges of $\Gamma(v)$ by $X \setminus \{x\}$.
        $Z^*_1 = \comp{Z^*_0 \cup \set{\set{w_2, w'_2}, \set{w_4, w'_4}}\setminus \set{\set{w_2, w_4}}}$ is the type-I neighbor with respect to $(\set{w_2,w_4},\set{w_2,w'_2},\set{w_4,w'_4})$. Then, we have $|Z^*_1 \cup Y| = 12 < 14 = |X \cup Y|$.}
        \label{fig:trans}
    \end{figure}
    
    Let $e_1, e_2, \cdots, e_t$ be the edges in $\we{X} \setminus Y$ in an arbitrary order, where $t = \size{\we{X} \setminus Y}$, and let $Z^*_0 = \comp{(X \setminus \set{x}) \cup (\we{X} \cup \set{e})}$ be the type-\II neighbor of $X$.
    For each $1 \le i \le t$, we define $Z^*_{i}$ as follows.
    If $e_i \notin Z^*_{i - 1}$, we define $Z^*_i = Z^*_{i-1}$.
    A concrete example of $Z^*_i$ can be found in~\Cref{fig:trans}. 
    
    Suppose otherwise that $e_i \in Z^*_{i-1}$.
    Let $w_i \in N(v)$ {be} one of the end vertices of $e_i$.
    Since $Y$ has no edges incident to $v$ and $e_i \notin Y$, there is an edge $f_i \in Y$ with $f_i \neq e_i$ that dominates edge $\{v, w_i\}$.
    Let $w'_i \neq w_i$ be the other end vertex of $e_i$ (i.e., $e_i = \{w_i, w'_i\}$). 
    If both $w_i$ and $w'_i$ are contained in $N(v)$, then $Y$ has an edge $h_i \in \Gamma(w'_i)\setminus \Gamma(v)$ as otherwise edge $\{v, w'_i\}$ is not dominated by $Y$, contradicting to the fact that $Y$ is an edge dominating set of $G$.
    We define $Z^*_i$ as the type-I neighbor of $Z^*_{i-1}$ with respect to $(e_i, f_i, h_i)$.
    Suppose next that $w'_i$ is not contained in $N(v)$.
    By the minimality of $\we{X}$, $X$ has no edges incident to $w_i$.
    This implies that $X$ contains an edge $h_i \in \Gamma(w'_i)$ dominating $e_i$.
    We then define $Z^*_i$ as the type-I neighbor of $Z^*_{i-1}$ with respect to $(e_i, f_i, h_i)$ if $h_i \neq x$, and otherwise, define $Z^*_i$ as the type-I neighbor of $Z^*_{i-1}$ with respect to $(e_i, f_i, e)$. 
    Note that if $h_i = x$, then one of the end vertices of $e_i$ is $u$, and then $Z^*_i$ is a legal type-I neighbor of $Z^*_{i-1}$.
    Finally, we have $Z^*_t$ and denote it by $Z$.
    Since $f_i \in Y$, $h_i \in X \cup Y$, $e_i \notin X \cup Y$ for every $1 \le i \le t$, and $e \in Y$,
    we have $Z \subseteq X \cup Y$.
    Moreover, as $x \notin Z$, we have $|X \cup Y| > |Z \cup Y|$.
\end{proof}

Finally, we analyze the delay of the algorithm.
Note that every minimal edge dominating set has $\order{n\Delta^2}$ neighbors. 
{To define the neighbors of a minimal edge dominating set $X$, we first choose an edge $e = \set{u, v} \in X$. The number of these choices is at most $n$ since $X$ has at most $n$ edges. For each $e$, we choose at most two edges in $\Gamma(u)\cup \Gamma(v)$. The number of such choices is at most  $\Delta^2$, and hence the number of neighbors of $X$ is $\order{n\Delta^2}$.}
Moreover, computing type-I and type-\II neighbors can be done in $\order{n\Delta}$ time. 
The details of this analysis can be found in the proof of the next theorem.

\begin{theoremrep}\label{theo:meds:poly}
    One can enumerate all minimal edge dominating sets in a graph with $\order{n^2\Delta^3}$ delay. 
\end{theoremrep}
\begin{proof}
    Each minimal edge dominating set of $G$ has $\order{n\Delta^2}$ neighbors and computing each neighbor can be done in $\order{n\Delta}$ time, assuming that the operation $\comp{\cdot}$ is done in time $O(n\Delta)$. 
    Thus, we consider the time complexity to compute $\comp{\tilde{X}}$. 
    Let $\tilde{X}$ be an edge dominating set of $G$.
    Since $\Pi_{\tt eds}$ is monotone property, we can compute $\comp{\tilde{X}}$ as the following procedure:
    \begin{enumerate}
        \item Let $(e_1, \ldots, e_{\size{\tilde{X}}})$ be the edges in $\tilde{X}$.
        \item For each $1 \le i \le \size{\tilde{X}}$, if $\tilde{X}\setminus \set{e_i}$ is an edge dominating set of $G$, then we update $\tilde{X}$ to $\tilde{X} \setminus\set{e_i}$.
    \end{enumerate}
    To check the condition that $\tilde{X} \setminus \{e_i\}$ is an edge dominating set of $G$, 
    it suffices to check whether $\size{(\Gamma(u) \cup \Gamma(v)) \cap (\tilde{X} \setminus \set{e_i})} \ge 1$ for any $\set{u, v} \in E$.
    This can be done in $O(\Delta)$ time for each Step 2 of the above procedure by simply keeping track of the value $\size{(\Gamma(u) \cup \Gamma(u)) \cap \tilde{X}}$ for each $\set{u, v} \in E$ and updating it in $O(\Delta)$ time when removing $e_i$ from $\tilde{X}$.
    Hence, we can compute $\comp{\tilde{X}}$ in $\order{\size{\tilde{X}}\Delta}$ time.
    Finally, we consider the cardinality of $\tilde{X}$ that appears in our algorithm. 
    In our algorithm, we use function $\comp{\cdot}$ when computing type-I or type-\II neighbors of a minimal edge dominating set $X$.
    By the definition of neighbors, $\tilde{X}$ is either of the form $(X \setminus \set{x}) \cup \{e, f\}$ or of the form $(X \setminus \set{x}) \cup (\we{X} \cup \{e\})$.
    For the former case, we have $\size{\tilde{X}} \le \size{X} + 2$ and for the latter case, we have $\size{\tilde{X}} \le \size{X} + \size{\we{X}}$.
    Since $\size{\we{X}}$ is at most $\Delta$, which is observed at the definition of type-\II neighbors, we have $\size{\tilde{X}} \le \size{X} + \Delta$. 
    Moreover, every minimal edge dominating set $X$ is a star forest.
    To see this, it suffices to show that any minimal edge dominating set contains neither a cycle $C_3$ of three vertices nor a path $P_4$ of four vertices as a subgraph. 
    If $X$ contains a $C_3$, then $X \setminus \set{e}$ is also an edge dominating set for any $e \in E(C_3)$.
    If $X$ contains a $P_4 = (e_1, e_2, e_3)$, then $X \setminus \set{e_2}$ is also an edge dominating set, contradicting the minimality of $X$.
    Thus, $X$ is a star forest and hence $|X| \le n - 1$.
\end{proof}

Since the best known delay of enumerating minimal edge dominating set is $\order{n^6}$ due to~\cite{DBLP:conf/wads/KanteLMNU15}, 
our result even improves the delay of minimal edge dominating set enumeration {since our algorithm runs in $\order{n^2\Delta^3} = \order{n^5}$ delay}. 


In order to extend the above enumeration algorithm to an approximate one for {weight-constrained} minimal edge dominating sets, we slightly modify the algorithm.
According to the definition of the type-\II neighbor, we can arbitrarily choose a minimal edge set $\we{X}$ such that $(X \setminus \set{x}) \cup (\we{X} \cup \set{e})$ is an edge dominating set of $G$ with $\we{X} \cap \Gamma(v) = \emptyset$.
However, for our approximate enumeration, $\we{X}$ must have small weight.
We observe that for every minimal edge dominating set $Y$ of $G$ with $Y \cap \Gamma(v) = \emptyset$, there is an edge set $X^*$ such that $(X \setminus \set{x}) \cup (X^* \cup \set{e})$ is an edge dominating set of $G$, $X^* \cap \Gamma(v) = \emptyset$, and $w(X^*) \le w(Y)$.
This follows from the fact that $(X \setminus \set{x}) \cup (Y \cup \set{e})$ is indeed an edge dominating set of~$G$ and $Y \cap \Gamma(v) = \emptyset$.

To compute a small weight edge set $X^*$ such that $(X \setminus \set{x}) \cup (X^* \cup \set{e})$ is an edge dominating set of $G$ and $X^* \cap \Gamma(v) = \emptyset$, we use a polynomial-time approximation algorithm $\mathcal A$ for the minimum weight edge dominating set problem.
Consider an edge-weighted graph $H = \bigcup_{u' \in U'}\Gamma(u')$, where $U'$ is set of vertices that has incident edges that are not dominated by $(X \setminus \set{x}) \cup \set{e}$.
For each edge $e \in \Gamma(v) \cap E(H)$, we set $w'(e) := \infty$ and for each edge $e \in E(H) \setminus \Gamma(v)$, we set $w'(e) := w(e)$.
It is not difficult to prove that $Y \cap E(H)$ is an edge dominating set of $H$. 
We apply $\mathcal A$ to $(H, w')$ and have an edge set $X'$ such that $(X \setminus \set{x}) \cup (X' \cup \set{e})$ is an edge dominating set of $G$ and $X' \cap \Gamma(v) = \emptyset$.
Moreover, we have $w(X^*) \le w(X') \le c\cdot w(Y)$, where $c$ is the approximation factor of $\mathcal A$.

The above lemma implies that every minimal edge dominating set on a suitable path on $\mathcal{G}$ from $X$ to $Y$ does not have weight more than $w(X) + (c + 1)w(Y)$. 
Note that we can compute $\we{X}$ in polynomial time with approximation factor $c = 2$ by a known approximation algorithm~\cite{Fujito:approximating:2017}.
Hence, the main theorem of this subsection is established. 

\begin{theoremrep}\label{theo:apx:eds}
    There is a polynomial-delay 5-approximate enumeration algorithm for enumerating minimal edge dominating sets using polynomial-time preprocessing.
\end{theoremrep}

\begin{proof}
    Let $X$ be an arbitrary minimal edge dominating set of $G$ with weight at most $ck$.
    We can find $X$ in time polynomial by~\cite{Fujito:approximating:2017} with $c = 2$.
    For every minimal edge dominating set $Y$ of weight at most $k$ distinct from $X$, there is a type-I neighbor $Z$ with $|Z \cup Y| < |X \cup Y|$ or a type-\II neighbor $Z^*_0$ such that there is a directed path $\mathcal P = (Z^*_0, Z^*_1, \ldots, Z^*_t = Z)$ from $Z^*_0$ to $Z$ in $\mathcal G$ with $|Z \cup Y| < |X \cup Y|$.
    Recall that, in the construction of $Z^*_i$, $Z^*_i = \comp{(Z^*_{i-1} \setminus \set{e_i}) \cup \set{f, h}}$ for some $e_i \in Z^*_{i-1} \setminus (X \cup Y)$ and $f, h \in X \cup Y$.
    It follows that $w(Z_i^* \cup X \cup Y) \le w(Z^*_{i-1} \cup X \cup Y)$ for every $1 \le i \le t$.
    Thus, for every $0 \le i \le t$,
    \begin{linenomath}
        \begin{align*}
            w(Z^*_{i} \cup Y) &\le w(Z^*_{0} \cup X \cup Y)\\
            &\le w((X \setminus \set{x}) \cup (\we{X} \cup \set{e}) \cup X \cup Y)\\
            &\le w(X \cup Y) + 2 \cdot w(Y).
        \end{align*}
    \end{linenomath}
    Since we compute $\we{X}$ by a polynomial-time $2$-approximation algorithm of \cite{Fujito:approximating:2017}, we have $\we{X} \le 2\cdot w(Y)$.
    Thus, every ``internal'' minimal edge dominating set $Z^*_i$ satisfies $w(Z^*_i \cup Y) \le w(X \cup Y) + 2\cdot w(Y) \le 5k$.
    Therefore, we can eventually find such a minimal edge dominating set $Z$ by traversing type-I or type-\II neighbors of weight at most $5k$.
    By applying the same argument from $Z$ to $Y$, every ``internal'' minimal edge dominating set~$Z'$, we have $w(Z' \cup Y) \le w(Z \cup Y) + 2 \cdot w(Y) < w(X \cup Y) + 2 \cdot w(Y) \le 5k$. 
    Therefore, by generating a type-I or a type-\II neighbors of weight at most~$5k$, we can, in polynomial delay, enumerate all the minimal edge dominating sets of $G$ of weight at most $k$ with approximation factor $5$. 
\end{proof}

\subsection{Minimal Steiner subgraphs}
Let $G = (V, E)$ be a graph and let $W \subseteq V$ be a set of terminals.
In this subsection, we assume that $|W| \ge 2$.
A subgraph $H$ of $G$ is a \emph{Steiner subgraph} of $(G, W)$ if there is a path between every pair of vertices of $W$ in $H$.
It is easy to see that the property of being a Steiner subgraph of $(G, W)$ is a monotone property over $E$.
Moreover, every minimal Steiner subgraph forms a tree, called a \emph{Steiner tree}.
There are polynomial-delay enumeration algorithms for minimal Steiner trees~\cite{DBLP:journals/is/KimelfeldS08,kobayashi2020polynomial}. 
These algorithms are based on branching and it seems not to be easy to extend these algorithms to an approximate enumeration algorithm.
Therefore, the goal of this subsection is to enumerate all the minimal Steiner trees of $(G, W)$ in polynomial delay with an approximation guarantee.

Let $\Pi_{\tt ST}$ be the property of being a minimal Steiner subgraph of $(G, W)$ and $w: E \to \mathbb Q_{> 0}$ be a weight function.
In this subsection, we may not distinguish between a subgraph and the set of edges in it. 
Unfortunately, the property $\Pi_{\tt ST}$ is not a CKS property.
To see this, we consider a graph with two adjacent terminals $s$ and $t$.
Let $x$ be the edge between $s$ and $t$.
Then, $X = \{x\}$ is a minimal Steiner tree of this graph and every other minimal Steiner tree is an $s$-$t$ path avoiding $x$ between $s$ and $t$ in the graph and hence there {may be} exponentially many such paths in general.
Thus, $\Pi_{\tt ST}$ is not a CKS property.
However, we can solve \IRP{\Pi_{\tt ST}}{w}{k} in polynomial delay since $\mathcal S(E, \Pi_{\tt}, w, k, X, x)$ contains all {paths of weight} at most~$k$ between the two components of $X \setminus \{x\}$, which can be enumerated in polynomial delay by the $k$-best enumeration algorithm for $s$-$t$ paths~\cite{Lawler1972}.
Moreover, since the minimum Steiner tree problem is polynomial-time approximable with factor $1.39$~\cite{Byrka:Steiner:2013}, by plugging these into~\Cref{theo:apx:output},
we obtain an incremental-polynomial time $3.39$-approximate enumeration algorithm for minimal Steiner subgraphs.
In this subsection, we further improve this result by giving a polynomial-delay $2.39$-approximate enumeration algorithm for minimal Steiner subgraphs.

The technique we used here is also the supergraph technique.
To this end, we define the neighbors of a minimal Steiner tree $X \subseteq E$ of $(G, W)$ in the supergraph $\mathcal G$ as follows.
We extend function $\comp{\cdot}$ in such a way that for a vertex set $W' \subseteq V$ and a Steiner subgraph $X$ of $(G, W')$, $\comp{X, W'}$ is a minimal Steiner tree of $(G, W')$ that is a subgraph of $X$.
Let $x$ be an edge in $X$.
Since $X$ is a tree, there are exactly two components $C_1$ and $C_2$ of $X \setminus \set{x}$.
Recall that $C_1$ and $C_2$ are also considered as sets of edges.
Let $e = \set{u, v}$ be an edge incident to a vertex in $C_1$ with $V(C_1) \cap \{u, v\} = \{u\}$ and $e \neq x$. 
Let $C'_1 = \comp{C_1 \cup \set{e}, W_1 \cup \{v\}}$ and $C'_2 = \comp{C_2, W_2}$, where $W_i \coloneqq W \cap V(C_i)$.  
Since $C_1 \cup \set{e}$ and $C_2$ are respectively Steiner subgraphs of $(G, W_1 \cup \{v\})$ and $(G, W_2)$, $C'_1$ and $C'_2$ are well-defined.
Let $P_{x, e}$ be an arbitrary shortest $v$-$w$ path in $G[V \setminus V(C'_1)]$, where $w$ is a vertex in $C'_2$. 
Let us note that $P_{x, e}$ might be a single vertex $v$ when $v = w$.
Observe that $C'_1 \cup C'_2 \cup P_{x, e}$ is a Steiner subgraph of $(G, W)$.
This follows from the fact that $P_{x, e}$ connects two components $V(C'_1)$ and $V(C'_2)$.
Then, $Z = \comp{C'_1 \cup C'_2 \cup P_{x, e}}$ is defined to be a neighbor of $X$ (with respect to $(x, e)$).
The neighbor of $X$ is defined to be the union of all neighbors with respect to possible pairs $x \in X$ and $e \in E\setminus \{x\}$.
{Each neighbor of $X$ is defined by two edges $e \in X$ and $f \in E \setminus X$.
Since $X$ has at most $n$ edges and $E \setminus X$ has $m$ edges,}
the number of neighbors of each minimal Steiner tree $X$ is $\order{nm}$, and we can enumerate all neighbors of $X$ in polynomial time. 

If $\mathcal G$ is strongly connected, then we can enumerate all minimal Steiner trees in polynomial delay. 
We begin with proving that $\mathcal G$ is strongly connected.
Before proving the strong connectivity of $\mathcal G$, we observe the following auxiliary lemma, which are easy to verify.

\begin{lemmarep}\label{lem:ir:st}
    Let $X$ and $Y$ be a pair of minimal Steiner trees of $(G, W)$, $x \in X$, and $C_1$ and $C_2$ be the two components of $X \setminus \{x\}$.
    Then, $Y$ has a path $P$ that connects $C_1$ and $C_2$.
\end{lemmarep}

\begin{proof}
    Since $X$ is a minimal Steiner tree of $(G, W)$, $C_1$ and $C_2$ have terminals $w_1$ and $w_2$, respectively.
    As $Y$ is also a Steiner subgraph of $(G, W)$, $Y$ has a $w_1$-$w_2$ path $P$.
    Thus, $P$ contains a subpath between vertices in $C_1$ and in $C_2$.
\end{proof}

The idea to proving the strong connectivity is as follows.
By~\Cref{lem:ir:st}, there is a path $P$ between $C_1$ and $C_2$.
If the shortest path $P_{x, e}$ is equal to $P$, then we are done: we can find a neighbor $Z \subseteq (X \setminus \set{e}) \cup P$ of $X$, which implies $\size{Z \cup Y} < \size{X \cup Y}$.
Otherwise, by appropriately selecting edges for a neighbor $Z'$ of $Z$, $Z'$ contains ``more edges'' of $P$ than $Z$ in some sense.
By repeating this, we can eventually find a neighbor $Z^*$ satisfying $Z^* \subseteq (X \setminus \set{e}) \cup P$, proving that $\size{Z^* \cup P} \le \size{Z^* \cup Y} < \size{X \cup Y}$.
The key to the strong connectivity of $\mathcal G$ is the following lemma.

\begin{lemmarep}\label{lem:st:scc}
    Let $X$ and $Y$ be a pair of minimal Steiner trees of $(G, W)$.
    Then, $\mathcal G$ has a directed path from $X$ to a minimal Steiner tree $Z$ that satisfies $\size{Z \cup Y} < \size{X \cup Y}$.
\end{lemmarep}

\begin{proof}
    Let $x_1$ be an edge in $X\setminus Y$ and let $C_1$ and $C_2$ be the components of $X \setminus \set{x_1}$.
    By~\Cref{lem:ir:st}, there is a path $P$ in $Y$ between vertices in $C_1$ and in $C_2$.
    Then, $(X \setminus \{x_1\}) \cup P$ is a Steiner subgraph of $(G, W)$.
    In the following, we show that there is a path from $X$ to a minimal Steiner tree $Z \subseteq E$ of $(G, W)$ in $\mathcal G$ with $Z \subseteq (X \setminus \{x_1\}) \cup P$, implying that
    \begin{linenomath}
        \begin{align*}
            \size{Z \cup Y} \le \size{(X \setminus \set{x_1}) \cup P \cup Y)} < \size{X \cup Y}.
        \end{align*}
    \end{linenomath}
    
    Let $e_1, \ldots, e_t$ be the sequence of edges of $P$ appearing in this order such that $e_1$ and $e_t$ are incident to a vertex in $C_1$ and a vertex in $C_2$, respectively.
    For each $1 \le i \le t$, we let $e_i = \{v_i, v_{i+1}\}$.
    As $e_i \in Y$, we have $x_1 \neq e_i$ for $1 \le i \le t$.
    We consider the neighbor $X_1$ of $X$ with respect to $(x_1, e_1)$. 
    Let $C'_1 = \comp{C_1 \cup \{e\}, W_1 \cup \{v\}}$ and $C'_2 = \comp{C_2, W_2}$, where $W_i = W \cap V(C_i)$.
    If $X_1 \subseteq (X \setminus \{x_1\}) \cup P$, we are done.
    Assume otherwise.
    Let $P_{x_1, e_1}$ be the shortest path used in defining $X_1$.
    By the assumption, we have $P_{x_1, e_1} \neq P$.
    Since $P_{x_1, e_1}$ contains $e_1$, $P$ and $P_{x_1, e_1}$ have common edge $e_1$.
    Let $e_{i_1}$ be the edge in $P$ such that for every $1 \le i \le i_1$, $e_i \in P_{x_1, e_1}$ and $e_{i_1} \notin P_{x_1, e_1}$.
    By the definition of $P_{x_1, e_1}$, we have $i_1 > 1$.
    Let $x_2 \neq e_{i_1 - 1}$ be the edge of $P_{x_1, e_1}$ that shares end vertex $v_{i_1}$ with $e_{i_1}$.
    Then, observe that the neighbor $X_2$ of $X_1$ with respect to $(x_2, e_{i_1})$ has all edges $e_1, \ldots, e_{i_1}$.
    This follows from the fact that we define $P_{x_2, e_{i_1}}$ as a shortest path between the two components of $C''_1 = \mu((X_1 \setminus \set{x_2}) \cup \{e_{i_1}\}, W_1 \cup \set{v_{i_1+1}})$ and $C''_2 = \mu(X_1\setminus\set{x_2}, W_2)$ and $v_{i_1+1}$ is considered as a terminal, which implies that $e_1, \ldots, e_{i_1}$ are all bridges in $C''_1$ and thus these edges are contained in $X_2$.
    Moreover, we claim that $X_2 \setminus C_1 \cup C_2 \cup \{e_1, \ldots, e_{i_1}\} \cup P_{x_2, e_{i_1}}$.
    This follows from the following facts: $C''_1 \subseteq C_1 \cup \{e_1, \ldots, e_{i_1}\}$; $C''_2 \subseteq C_2$. 

    By repeatedly applying this argument, we can eventually find a minimal Steiner tree $Z$ that contains all edges in $P$ and hence we have $Z \subseteq C_1 \cup C_2 \cup \{e_1, \ldots, e_t\} \subseteq (X \setminus \{x_1\}) \cup P$.
\end{proof}

Thus, {$\mathcal G$ is strongly connected and}
we can enumerate all minimal Steiner trees in polynomial delay. 
Next, we prove that we can approximately enumerate all minimal Steiner trees with the approximation factor $c + 1$.
Suppose that we have a polynomial-time $c$-approximation algorithm for finding a minimum weight Steiner tree.
Let $X$ be a minimal Steiner tree of $(G, W)$ whose weight is at most $ck$.
By \Cref{lem:st:scc}, $\mathcal G$ has a directed path from $X$ to every minimal Steiner tree $Y$ with weight at most $k$.
To show that we can approximately enumerate all solutions, we prove that every internal node on the directed path has weight at most $(c + 1)k$.

\begin{lemmarep}\label{lem:scc:apx:st}
    Let $X$ and $Y$ be a pair of minimal Steiner tree of $(G, W)$ with weight at most $ck$ and $k$, respectively.
    Then, $\mathcal G$ has a directed path from $X$ to $Y$ on which any minimal Steiner tree has weight at most $(c + 1)k$.
\end{lemmarep}
\begin{proof}
    Let $Z$ be a minimal Steiner tree of $(G, W)$ defined in \Cref{lem:st:scc}.
    Then, $\mathcal G$ has a directed path from $X$ to $Z$. 
    In the proof of \Cref{lem:st:scc}, we have shown that $X_j \subseteq C_1 \cup C_2 \cup \{e_1, \ldots, e_{i_j-1}\} \cup P_{x_j, e_{i_j}}$ for every minimal Steiner tree $X_j$ on the path from $X$ to $Z$.
    Since $P_{x_j, e_{i_j}}$ is chosen as a shortest path, its weight is at most $w(\{e_{i_j - 1}, \ldots, e_t\})$.
    Thus, we have
    \begin{linenomath}
        \begin{align*}
            w(X_j) &\le w(C_1 \cup C_2 \cup \{e_1, \ldots, e_{i_j-1}\} \cup P_{x_j, e_{i_j}}) \\
            &< w(X \cup P) \\
            &\le w(X \cup Y) \\
            &\le (c+1)k.
        \end{align*}
    \end{linenomath}
    By applying \Cref{lem:st:scc} again for $Z$ and $Y$, every solution $Z'$ from $Z$ to $Y$, we also have
    \begin{linenomath}
        \begin{align*}
            w(Z') \le w(Z \cup Y) < w(X \cup Y) \le (c+1)k,
        \end{align*}
    \end{linenomath}
    as $Z \subseteq X \cup Y$, completing the proof of lemma.
\end{proof}

Note that there can be a directed path from $X$ to $Y$ in $\mathcal G$ that has an internal node corresponding to a minimal Steiner tree with weight more than $(c+1)k$.
In our enumeration algorithm, we simply ignore such solutions and, by~\Cref{lem:st:scc,lem:scc:apx:st}, the ``trimmed'' supergraph still has a directed path from any solution with weight at most $ck$ to any solution with weight at most~$k$.
Therefore, we obtain the following theorem.

\begin{theorem}\label{theo:apx:st}
    There is a polynomial-delay $2.39$-approximate enumeration algorithm for enumerating minimal Steiner trees using polynomial-time preprocessing {since a polynomial-time approximation algorithm with approximation factor $1.39$ is known~\cite{Byrka:Steiner:2013}}.
\end{theorem}

{
\section{Concluding remarks}
    In this paper, we propose a new concept, approximate enumeration.
    \revise{For a combinatorial optimization problem,}
    an enumeration algorithm approximately enumerates $\mathcal S$ if it enumerates a set of feasible solutions $\mathcal R'$ with $\mathcal S \subseteq \mathcal R' \subseteq \mathcal R$ without duplication\revise{, where $\mathcal R$ is the collection of all feasible solutions and $\mathcal S$ is the subcollection of feasible solutions of weight at most $k$.}
    In other words, we allow enumeration algorithms to output some feasible solutions in $\mathcal R \setminus \mathcal S$ but forbid them to output infeasible solutions in $2^{U} \setminus \mathcal R$. 
}

{
    As technical contributions, 
    we show that the input-restricted problem is also essential for 
    designing efficient approximate enumeration algorithms for minimal subsets.
    The input-restricted problem can be solved efficiently with a constant approximation factor in order to achieve approximate enumeration with a constant approximation factor.
}

{
    As results beyond this framework, 
    minimal edge dominating sets and minimal Steiner trees can be enumerated approximately with polynomial delay.
    These properties are not CKS properties.
    Therefore, we cannot obtain polynomial-delay enumeration algorithms with a simple application of our framework. 
    However, by carefully defining the neighborhood of each solution,
    we obtain an enumeration algorithm with polynomial delay and a constant approximation factor.
}

{
    As for future research, we consider two possible directions. One direction is to enumerate \emph{exactly} all the minimal subsets of a certain property $\Pi$ that satisfy a weight constraint. 
    This immediately requires a polynomial-time algorithm for finding a minimum set satisfying $\Pi$.
}

{
    Another direction is to enumerate all \revise{\emph{maximal}} subsets with weight constraints approximately.
    Our technique cannot be extended to enumerate all maximal subsets in a straightforward way. 
    In the case of minimal subset enumeration, the cardinality of the union of two ``small'' minimal subsets is also small.
    On the other hand, in the case of maximal subset enumeration,
    the cardinality of the intersection of two ``large'' maximal subsets can be ``small''.
    Thus, our approach cannot be used in this case and
    it would be interesting to devise a new technique to overcome this obstacle.
}

\section*{Acknowledgement}
This work was partially supported by JST, CREST Grant Number JPMJCR18K3, Japan and 
JSPS Kakenhi Grant Numbers JP19H01133, JP19K20350, JP20K19742, JP20H05793, and JP21K17812.

\bibliographystyle{plain}
\bibliography{main}

\end{document}